\documentclass[conference]{IEEEtran}
\IEEEoverridecommandlockouts
\usepackage{cite}
\usepackage{amsmath,amssymb,amsfonts}
\usepackage{comment}
\usepackage{algorithmic}
\usepackage{graphicx}
\usepackage{textcomp}
\usepackage{xcolor}

\def\BibTeX{{\rm B\kern-.05em{\sc i\kern-.025em b}\kern-.08em
    T\kern-.1667em\lower.7ex\hbox{E}\kern-.125emX}}

\def \Pih{\widehat{\Pi}}

\def \thetah{\hat{\theta}}

\def \Pit{\widetilde{\Pi}}
\def \TV {\mathrm{TV}}
\def \Tr{\mathrm{Tr}}
\def \Rot{\mathrm{Rot}}
\newcommand{\bra}[1]{\langle#1\rvert} 
\newcommand{\ket}[1]{\lvert#1\rangle} 

\title{Transfer Learning for Quantum Classifiers: An Information-Theoretic Generalization Analysis\\
\thanks{STJ is with the Department of Computer Science, University of Birmingham (email:s.t.jose@bham.ac.uk) and OS is with the Department of Engineering, King's College London (email:  osvaldo.simeone@kcl.ac.uk). The work was done when STJ was a PostDoc at King's College London. STJ and OS received funding from the European Research Council
(ERC) under the European Union’s Horizon 2020 Research and Innovation
Programme (Grant Agreement No. 725731). OS was also supported by an open Fellowship of the EPSRC (EP/W024101/1). The authors would like to thank Dr. Ivana Nikoloska for useful discussions in the early stages of this work.}
}

\author{\IEEEauthorblockN{ Sharu Theresa Jose and Osvaldo Simeone}
}

\usepackage{amsmath, amssymb, bbm, xspace}
\usepackage{epsfig}
\usepackage{longtable}
\usepackage{color}
\usepackage{mathrsfs}
\usepackage{comment}

\usepackage{courier}

\newtheorem{theorem}{Theorem}[section]

\newtheorem{lemma}{Lemma}[section]

\newtheorem{definition}{Definition}[section]

\def\bkE{{\rm I\kern-.17em E}}
\def\bk1{{\rm 1\kern-.17em l}}
\def\bkD{{\rm I\kern-.17em D}}
\def\bkR{{\rm I\kern-.17em R}}
\def\bkP{{\rm I\kern-.17em P}}

\def\bkZ{{\bf{Z}}}

\def\bkE{{\rm I\kern-.17em E}}
\def\bk1{{\rm 1\kern-.17em l}}
\def\bkD{{\rm I\kern-.17em D}}
\def\bkR{{\rm I\kern-.17em R}}
\def\bkP{{\rm I\kern-.17em P}}

\makeatletter
\newcommand{\pushright}[1]{\ifmeasuring@#1\else\omit\hfill$\displaystyle#1$\fi\ignorespaces}
\newcommand{\pushleft}[1]{\ifmeasuring@#1\else\omit$\displaystyle#1$\hfill\fi\ignorespaces}
\makeatother

\def\bkZ{{\bf{Z}}}
\def\b12{(\beta_1,\beta_2)}

\newcounter{example}
\renewcommand{\theexample}{\thesection.\arabic{example}}

\newcounter{remark}
\renewcommand{\theremark}{\thesection.\arabic{remark}}

\def\Xscr{\mathcal{X}}

\def\Ebb{\mathbb{E}}
\newlength{\noteWidth}
\long\def\notes#1{\ifinner
{\tiny #1}
\else
\marginpar{\parbox[t]{\noteWidth}{\raggedright\tiny #1}}
\fi\typeout{#1}}

 \def\notes#1{\typeout{read notes: #1}} 

\newcommand{\ie}{i.e.\@\xspace} 

\newcommand{\Real}{\ensuremath{\mathbb{R}}}

\def\Ebb{\mathbb{E}}

\def\exp{\mathop{\hbox{\rm exp}}}

\def\spose#1{\hbox to 0pt{#1\hss}}

\def\text #1{\hbox{\quad#1\quad}}

\def\Escr{\mathcal{E}}

\def\nthinsp{\mskip -2   mu}

\def\superstar{^{\raise 0.5pt\hbox{$\nthinsp *$}}}
\def\SUPERSTAR{^{\raise 0.5pt\hbox{$*$}}}

\def\lamstarT {\lambda^{\raise 0.5pt\hbox{$\nthinsp *$}T}}

\def\Ascr{{\cal A}}

\def\Fscr{{\cal F}}
\def\Dscr{{\cal D}}

\def\Mscr{{\cal M}}
\def\Tscr{{\cal T}}

\def\Oscr{{\cal O}}

\def\Sscr{{\cal S}}

\def\Uscr{{\cal U}}

\def\Mscr{{\cal M}}
\def\Nscr{{\cal N}}
\def\Rscr{{\cal R}}
\def\Gscr{{\cal G}}
\def\Cscr{{\cal C}}

\def\Xscr{{\cal X}}

\def\non{\nonumber}

\let\forallnew\forall
\renewcommand{\forall}{\forallnew\ }
\let\forall\forallnew

		\def\bkE{{\rm I\kern-.17em E}}
		\def\bk1{{\rm 1\kern-.17em l}}
		\def\bkD{{\rm I\kern-.17em D}}
		\def\bkR{{\rm I\kern-.17em R}}
		\def\bkP{{\rm I\kern-.17em P}}
		\def\bkY{{\bf \kern-.17em Y}}
		\def\bkZ{{\bf \kern-.17em Z}}
		\def\bkC{{\bf  \kern-.17em C}}

{\begin{list}{}%
         {\setlength{\leftmargin}{#1}}%
         \item[]%
}
{\end{list}}

		\def\bsp{\begin{split}}
		\def\beq{\begin{eqnarray}}
		\def\bal{\begin{align*}}
		\def\bc{\begin{center}}
		\def\be{\begin{enumerate}}
		\def\bi{\begin{itemize}}
		\def\bs{\begin{small}}
		\def\bS{\begin{slide}}
		\def\ec{\end{center}}
		\def\ee{\end{enumerate}}
		\def\ei{\end{itemize}}
		\def\es{\end{small}}
		\def\eS{\end{slide}}
		\def\eeq{\end{eqnarray}}
		\def\eal{\end{align*}}
		\def\esp{\end{split}}
		\def\qed{ \vrule height7.5pt width7.5pt depth0pt}  

	\def\cp2problem#1#2#3#4{\fbox
		 {\begin{tabular*}{0.9\textwidth}
			{@{}l@{\extracolsep{\fill}}l@{\extracolsep{6pt}}l@{\extracolsep{\fill}}c@{}}
				#1 & & $#4 $ 
			\end{tabular*}}}

		\def\bkE{{\rm I\kern-.17em E}}
		\def\bk1{{\rm 1\kern-.17em l}}
		\def\bkD{{\rm I\kern-.17em D}}
		\def\bkR{{\rm I\kern-.17em R}}
		\def\bkP{{\rm I\kern-.17em P}}
		
		\def\bkZ{{\bf{Z}}}

\newcommand {\beeq}[1]{\begin{equation}\label{#1}}
\newcommand {\eeeq}{\end{equation}}
\newcommand {\bea}{\begin{eqnarray}}
\newcommand {\eea}{\end{eqnarray}}

\def\texitem#1{\par\smallskip\noindent\hangindent 25pt
               \hbox to 25pt {\hss #1 ~}\ignorespaces}

\def\bsp{\begin{split}}
		\def\beq{\begin{eqnarray}}
		\def\bal{\begin{align*}}
		\def\bc{\begin{center}}
		\def\be{\begin{enumerate}}
		\def\bi{\begin{itemize}}
		\def\bs{\begin{small}}
		\def\bS{\begin{slide}}
		\def\ec{\end{center}}
		\def\ee{\end{enumerate}}
		\def\ei{\end{itemize}}
		\def\es{\end{small}}
		\def\eS{\end{slide}}
		\def\eeq{\end{eqnarray}}
		\def\eal{\end{align*}}
		\def\esp{\end{split}}
		\def\qed{ \vrule height7.5pt width7.5pt depth0pt}  

\usepackage{amsmath, amssymb, xspace}
\usepackage{epsfig}
\usepackage{longtable}
\usepackage{color}
\usepackage{mathrsfs}
\usepackage{subfig}
\newenvironment{proof}[1][]{{\noindent \textit{ Proof}: }}{\hfill \qed \vspace{3pt}\\ }

\def\Cscr{{\cal C}}
\def\Nscr{{\cal N}}

\begin{document}
\maketitle

\begin{abstract}
A key component of a quantum machine learning model operating on classical inputs is the design of an embedding circuit mapping inputs to a quantum state. This paper studies a transfer learning setting in which classical-to-quantum embedding is carried out by an arbitrary parametric quantum circuit that is pre-trained based on data from a source task. At run time, a binary quantum classifier of the embedding is optimized based on data from the target task of interest. 
The average excess risk, i.e., the optimality gap, of the resulting classifier depends on how (dis)similar the source and target tasks are.
We introduce a new measure of (dis)similarity between the binary quantum classification tasks via the trace distances. An upper bound on  the optimality gap is  derived in terms of the proposed task (dis)similarity measure, two R\'enyi mutual information terms between classical input and quantum embedding under source and target tasks, as well as a measure of complexity of the combined space of quantum embeddings and classifiers under the source task. The theoretical results are  validated on a simple binary classification example. 
\end{abstract}

\section{Introduction}\label{sec:introduction}
Quantum machine learning (QML) is an emerging paradigm for programming noisy, intermediate scale quantum (NISQ) computers \cite{simeone2022introduction}. In QML, the parameter vector $\theta$ defining the operation of a 
parametric quantum circuit (PQC) is optimized based on quantum or classical data.
When input data are classical, it is necessary to  design an embedding circuit to map classical inputs to a quantum state \cite[Ch. 6]{schuld2021machine}. This is illustrated in Fig.~\ref{fig:QTL}, in which the classical input vector $x$ is mapped, via a PQC, to a quantum state defined by a density matrix $\rho_{\theta}(x)$. We focus on the task of classifying the input $x$ by applying a quantum measurement $\{M_c\}$ to the density matrix $\rho_{\theta}(x)$ \cite{schuld2021machine,helstrom1969quantum}. Following the QML framework, both the parameter vector $\theta$ of the embedding circuit and the classifying quantum measurement are optimized  based on supervised examples of the form $(c,x)$, where $c$ is a binary label. This paper analyzes the generalization properties of the trained quantum classifier.

\begin{figure}
    \centering
    \includegraphics[scale=0.34,trim=0.3in 1in 0in 0in,clip=true]{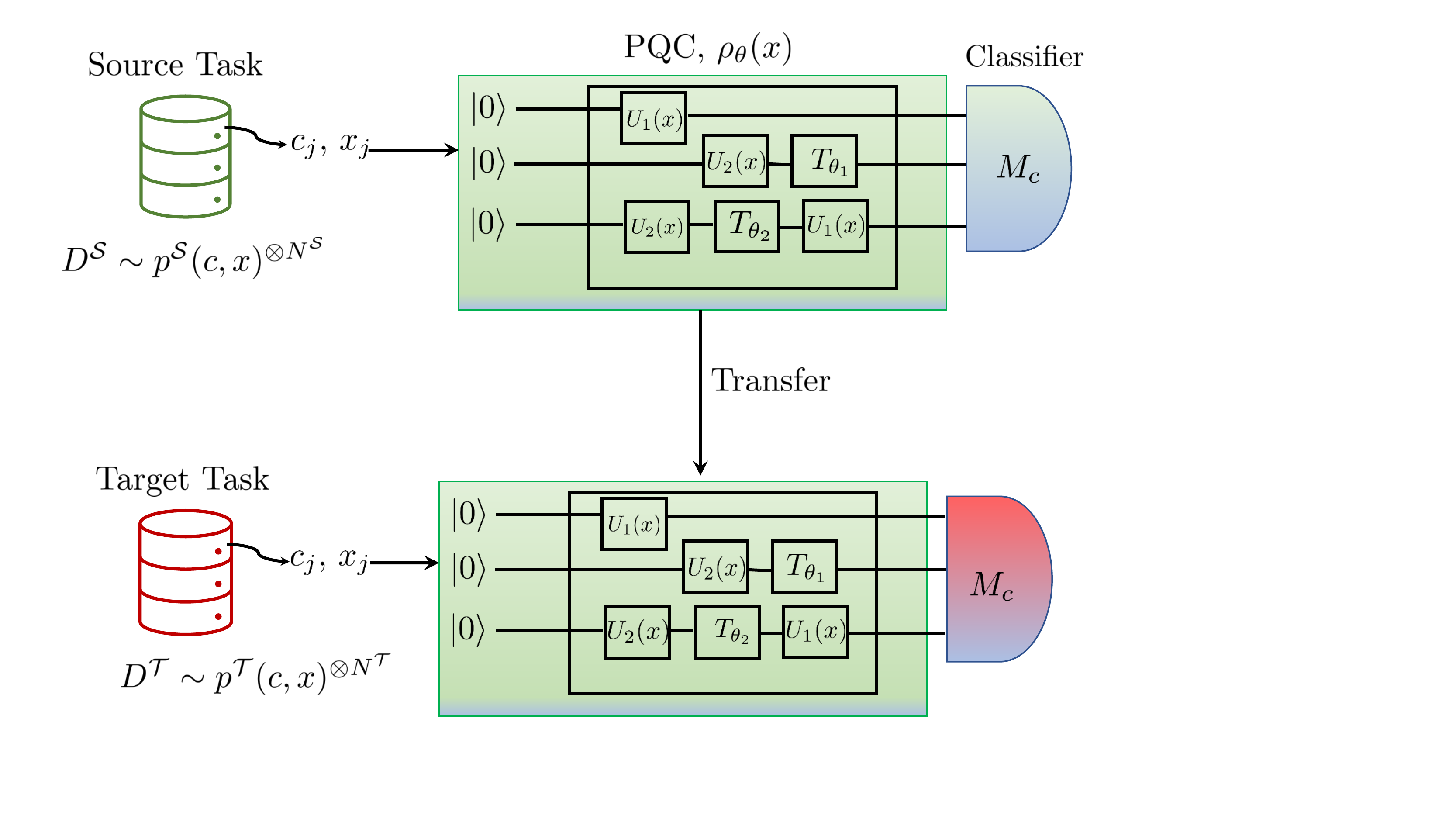}
    \caption{Illustration of the transfer learning problem under study. Data from the source task (green) is used to learn the parameters $\theta$ of the parameterized quantum circuit (PQC) implementing classical-to-quantum embedding which is then fixed for use by the target task. The target task uses its data (red) to learn the optimal task-specific binary classifier. }
    \label{fig:QTL}
    \vspace{-0.4cm}
\end{figure}
Reference \cite{banchi2021generalization} has recently studied the generalization properties of the circuit in Fig.~\ref{fig:QTL} for a \emph{fixed} embedding parameter $\theta$ as a function of the number $N$ of examples $(c,x)$ used to optimize the measurement from an information-theoretic viewpoint. The authors have shown that the excess risk, i.e., the optimality gap, can be bounded as $\Oscr(\sqrt{2^{I_2(X;R_{\theta})}}/\sqrt{N})$, where $I_2(X;R_{\theta})$ is the $2$-R\'enyi mutual information (MI) between the classical input $x$ and the quantum embedding $\rho_{\theta}(x)$ under the classical-quantum state $\rho_{XR_{\theta}}=\Ebb_{p(x)}[|x\rangle \langle x|\otimes \rho_\theta(x)]$, with $p(x)$ being the marginal of the ground-truth distribution $p(c,x)$ \cite{wilde2013quantum}.

In this paper, we consider the more practical case in which one needs to design \emph{both} embedding parameter vector $\theta$ and classifying measurement $\{M_c\}$. Furthermore, we address the challenging scenario in which  limited data is available from the target task. To this end, as in \cite{mari2020transfer}, we assume that the embedding circuit producing the state $\rho_\theta(x)$ is \emph{pre-trained} based on, generally more abundant, data from a related \emph{source task}. With the pre-trained embedding circuit $\rho_{\theta}(x)$, data from the target task is used only to adapt the classifying measurement. Therefore, 
differently from \cite{banchi2021generalization}, the average excess risk of the resulting classifier with respect to the target task depends crucially on how similar the source and target tasks are. 

Our main contribution is  a new measure of dissimilarity between binary quantum classification tasks that allows us to derive an upper bound on the average excess risk of the binary classifier. The derived bound scales as
 $\Oscr(\sqrt{\sup_{\theta}2^{I_2^{\Tscr}(X;R_{\theta})}}/\sqrt{N^{\Tscr}})+\Oscr((\mathfrak{R}^{\Sscr}_{\Theta,\Mscr}+\sqrt{\sup_{\theta}2^{I_2^{\Sscr}(X;R_{\theta})}})/\sqrt{N^{\Sscr}})+D^{ST}$, where $D^{ST}$ is the proposed (dis)similarity between the source and target tasks;  $I_2^{\Sscr}(X;R_{\theta})$ and $I_2^{\Tscr}(X;R_{\theta})$ are the $2$-R\'enyi MIs  between classical input and quantum embedding under the ground-truth data distributions of source and target tasks; and $\mathfrak{R}^{\Sscr}_{\Theta,\Mscr}$ is the Rademacher complexity of the joint space of quantum embeddings and measurements that scales with the dimension of the Hilbert space of the quantum embedding.

Apart from the mentioned reference \cite{banchi2021generalization}, generalization properties of variational quantum circuits as QML models have been characterized via an information geometric approach based on Fisher information \cite{abbas2021power}; via the Rademacher complexity of the space of PQCs measured in terms of the ratio of the number of gates in the circuit to the number of data samples \cite{caro2021generalization}; via a covering number based measure of expressivity of variational quantum circuits \cite{du2022efficient};  or as  a function  of data encoding strategies \cite{caro2021encoding}. Also related is reference \cite{tripuraneni2020theory}, which studies a classical version of the problem considered in this work, revealing the role of task similarity for transfer representation learning. To the best of our knowledge, ours is the first work that studies the generalization error incurred in transfer learning  quantum embeddings from an information-theoretic perspective.

The rest of the paper is organized as follows. Sec.\ref{sec:problem} details the two-stage transfer learning problem under study and defines the average excess risk. Sec.\ref{sec:similaritymetric} introduces a  similarity metric between source and target tasks based on trace distances. Leveraging this metric, Sec.\ref{sec:infobounds} presents an information-theoretic upper bound on the average excess risk. Theoretical conclusions are demonstrated via  examples in Sec.\ref{sec:example}.
\section{Problem Formulation}\label{sec:problem}
In this section, we first describe the  quantum classification problem studied in \cite{banchi2021generalization} in which the  quantum embedding parameter vector $\theta$ is fixed, and then we present the two-stage transfer learning problem illustrated in Fig.~\ref{fig:QTL} in which the embedding circuit parameter $\theta$ is pre-trained based on data from a separate source task.

\subsection{Quantum Classification with a Fixed Embedding}\label{sec:quantclassification_fixedembedding}

Let $x$ denote the classical input feature vector and $c \in \{0,1\}$ be the corresponding binary class index. We take $x$ to assume values in an arbitrary discrete finite set, although extensions to continuous-valued inputs are direct \cite{banchi2021generalization}. The data sample $(c,x)$ is generated from an \emph{unknown} underlying joint distribution $p^{\mathcal{T}}(c,x)$ describing the \emph{target} task. The \emph{embedding circuit} maps the classical feature vector $x$  to a \textit{density matrix} $\rho_{\theta}(x)$, i.e., to  a positive semi-definite unit-trace matrix defined on some (finite-dimensional) Hilbert space. The embedding circuit is implemented by a PQC parameterized by a (classical) parameter vector $\theta \in \Theta$, where $\Theta$ is an arbitrary set.

The \emph{classifier} consists of a positive operator-valued measure (POVM) applied to the quantum state $\rho_{\theta}(x)$. The POVM is defined by two positive-semidefinite matrices  $M=\{M_c\}_{c=0}^{1}$, of the same dimensions of the density matrix $\rho_{\theta}(x)$, that satisfy the conditions $M_c \geq 0$ and $\sum_{c=0}^{1}M_c=I$. By Born's rule, the classifier chooses class $c$ with probability $\Tr(M_c \rho_{\theta}(x))$, where $\Tr(\cdot)$ represents the trace operation. We use $\Mscr=\{M:M_c \geq 0, \sum_{c=0}^{1}M_c=I\}$ to denote the set of all binary POVMs for the given Hilbert space.

 For a fixed embedding parameter $\theta$, quantum supervised classification \cite{banchi2021generalization} optimizes the POVM $M \in \Mscr$ with the ideal goal of minimizing the expected probability of error, also known as the \emph{expected risk}, i.e.,
\begin{align}
    \Rscr^{\mathcal{T}}_{\theta,M}=\Ebb_{p^{\mathcal{T}}(c,x)}[\ell_{\theta,M}(c,x)], \label{eq:exp_risk}
\end{align}over $M \in \Mscr$, where \begin{align}
    \ell_{\theta,M}(c,x)=1-\Tr(M_c \rho_{\theta}(x)) \label{eq:loss}
\end{align} is the probability of error evaluated on an  example $(c,x)$. Accordingly, the minimum expected risk for parameter $\theta$ is given as
\begin{align}
    \Rscr^{\Tscr}_{\theta}= \min_{M \in \Mscr}\Rscr^{\mathcal{T}}_{\theta,M} \label{eq:minexpectedrisk}.
\end{align}

Since the ground-truth joint distribution $p^{\mathcal{T}}(c,x)$ is unknown, the optimization of the POVM $M$ is done by using a training data set $\mathcal{D}^{\Tscr}=\{(c_1,x_1), \hdots, (c_{N^{\Tscr}},x_{N^{\Tscr}})\}$ of $N^{\Tscr}$ samples, whose individual data points $(c_j,x_j)$ are assumed to be independent identically distributed (i.i.d.) according to $p^{\mathcal{T}}(c,x)$. Specifically, the POVM is obtained by minimizing the \emph{training loss}\begin{align}\label{eq:emptarget}
    \widehat{\Rscr}^{\mathcal{T}}_{\theta,M}=\frac{1}{N^{\mathcal{T}}}\sum_{(c,x)\in \mathcal{D}^{\Tscr}} \ell_{\theta,M}(c,x).
\end{align}The solution of this optimization can be obtained in closed form, yielding the so-called \emph{Hellstrom measurement} (see \cite[Sec. III]{helstrom1969quantum}). We write as  \begin{align}\label{eq:notation}
        \widehat{M}^{\Tscr}_{\theta}=\arg \min_{M \in \Mscr}\widehat{\Rscr}^{\Tscr}_{\theta,M} \textrm{ and } \widehat{\Rscr}^{\mathcal{T}}_{\theta}=\widehat{\Rscr}^{\mathcal{T}}_{\theta,\widehat{M}^{\Tscr}_{\theta}}
    \end{align}the  optimal POVM and the corresponding minimized training loss for a fixed $\theta$, respectively.

The classifier obtained with the POVM \eqref{eq:notation} is considered to generalize well if it yields a low expected risk (\ref{eq:loss}). In this regard, a key metric of interest is the \emph{excess risk} \begin{align}
    \Delta\Rscr_{\theta}^{\Tscr} =\Rscr^{\Tscr}_{ \theta,\widehat{M}^{\Tscr}_{\theta}}-\Rscr^{\Tscr}_{ \theta} \label{eq:excess_risk_1},
\end{align} which is the difference between the expected risk (\ref{eq:exp_risk}) obtained via the outlined learning process and the genie-aided expected risk obtained with the optimal POVM. As described in Section~\ref{sec:introduction}, an information-theoretic bound on the excess risk (\ref{eq:excess_risk_1}) was derived in  \cite{banchi2021generalization} for a fixed parameter $\theta$.
\subsection{Transfer Learning for Quantum Classification}
In this work, as illustrated in Fig.~\ref{fig:QTL}, we consider a two-stage transfer learning problem, in which the embedding parameter vector $\theta$ is pre-trained based on data from a \emph{source} task with underlying true data distribution $p^{\Sscr}(c,x)$, which is generally different from the distribution $p^{\Tscr}(c,x)$ of the target task. To this end, we assume to have access to a training set    $\mathcal{D}^{\Sscr}=\{(c_1,x_1), \hdots, (c_{N^{\Sscr}},x_{N^{\Sscr}})\}$ of $N^{\Sscr}$ samples generated i.i.d. according to the source task distribution $p^{\Sscr}(c,x)$. In a typical implementation, one uses source-task data to compensate for limitations in the availability of target-task data. Therefore,  one may assume that the number of data samples $N^{\Sscr}$ from the source task is larger than that for the target task, \ie, $N^{\Tscr} \ll N^{\Sscr}$.

As illustrated in Fig.~\ref{fig:QTL}, in the pre-training phase, the source-task data set $\mathcal{D}^{\Sscr}$ is used to optimize the embedding parameter $\theta$. In the training phase, the embedding parameter is fixed to the pre-trained value $\thetah$ obtained from the first phase, and the POVM for the target task is optimized as described in the previous subsection.

To elaborate, in the \emph{pre-training phase}, the embedding parameter vector $\theta$ is obtained by minimizing the
    \emph{training loss on the source-task data} as
    \begin{align}
        \thetah= \arg \min_{\theta \in \Theta} \widehat{\Rscr}^{\Sscr}_{\theta} \label{eq:thetahat},
    \end{align} where we have defined the source-task  training loss as $\widehat{\Rscr}^{\Sscr}_{\theta}=\arg \min_{M \in \Mscr} \widehat{\Rscr}^{\mathcal{S}}_{\theta,M}$ with
    $\widehat{\Rscr}^{\mathcal{S}}_{\theta,M}=\sum_{(c,x)\in \mathcal{D}^{\Sscr}} \ell_{\theta,M}(c,x)/N^{\Sscr}$.  In the \emph{training phase}, the classifying measurement is optimized as in \eqref{eq:notation} using the target-task data for the pre-trained embedding parameter vector $\thetah$ in (\ref{eq:thetahat}). This yields the POVM $\widehat{M}^{\Tscr}_{\thetah}$ and the expected risk $\widehat{\Rscr}^{\mathcal{T}}_{\thetah}=\widehat{\Rscr}^{\mathcal{T}}_{\thetah,\widehat{M}^{\Tscr}_{\thetah}}$.

In order to evaluate the generalization properties of transfer learning, we adopt the \emph{transfer excess risk} 
\begin{align}
    \Delta\Rscr^{\Sscr \rightarrow \Tscr} =\Rscr^{\Tscr}_{\thetah,\widehat{M}^{\Tscr}_{\thetah}}-\min_{\theta\in \Theta}\Rscr^{\Tscr}_{\theta} \label{eq:excess_risk}.
\end{align} Unlike the excess risk in  (\ref{eq:excess_risk_1}), the transfer excess risk captures the impact on generalization not only of the classifier, which is trained using target-task data, but also of the embedding parameter $\theta$, which is pre-trained using source-task data.
The transfer excess risk \eqref{eq:excess_risk} thus depends intuitively on how ``similar'' the embedding parameter vectors $\theta$ that minimize the losses on the source and target tasks are. 

\vspace{-0.2cm}
\section{On the Similarity of Source and Target Tasks}\label{sec:similaritymetric}
In this section, we present a similarity metric for source and target tasks that will be shown in the next section to determine a bound on the transfer excess risk \eqref{eq:excess_risk}. To this end, we start with some preliminary background on quantum information.
\vspace{-0.2cm}
\subsection{Preliminaries}
Let $\rho$ and $\sigma$ denote two square matrices defined on the same (finite-dimensional) Hilbert space. The \emph{trace distance} $T(\rho,\sigma)$ between the matrices $\rho$ and $\sigma$ 
is defined as \cite{wilde2013quantum} 
\begin{align}\label{eq:trdist}
   T(\rho,\sigma)=\frac{1}{2} \lVert \rho-\sigma\rVert_1 ,
\end{align}where $\lVert A \rVert_1=\Tr(\sqrt{A^{\dag}A})$ is the \emph{trace norm} of the matrix $A$, with $A^{\dag}$ denoting the conjugate transpose of $A$. The trace distance satisfies triangle inequality, and for two density matrices $\rho$ and $\sigma$, it is bounded as $0 \leq T(\rho,\sigma) \leq 1$ \cite{wilde2013quantum}.

\subsection{Task-Induced Distance between Embedding Parameters}
We start by  defining a distance measure $d^{\Ascr}(\theta,\theta')$ between two embedding parameter vectors $\theta$ and $\theta' \in \Theta$  induced by a task $\Ascr \in \{\Sscr,\Tscr\}$.  The distance $d^{\Ascr}(\theta,\theta')$ is given by the difference between the minimum expected risks \eqref{eq:minexpectedrisk} obtained with embedding parameters $\theta$ and $\theta'$.
\begin{definition}[Task-based Distance]\label{def:taskbaseddistance}
For task $\mathcal{A}$, with $\mathcal{A} \in \{\Sscr,\Tscr\}$, the \emph{task-based distance} between any two embedding parameters $\theta$ and $\theta' \in \Theta$ is defined as
\begin{align}
    d^{\mathcal{A}}(\theta,\theta')=|\Rscr^{\mathcal{A}}_{\theta'}-\Rscr^{\mathcal{A}}_{\theta}|, \label{eq:distance_measure}
\end{align}where the minimum expected risk is defined in \eqref{eq:minexpectedrisk}. We have the inequalities $0 \leq d^{\mathcal{A}}(\theta,\theta') \leq 0.5$.
\end{definition}

The task-based distance \eqref{eq:distance_measure} can be computed explicitly by introducing the \emph{class-$c$ average density matrix}
 \begin{equation}\label{eq:perdens}\rho^{\mathcal{A}}_{\theta|c}=\Ebb_{p^{\mathcal{A}}(x|c)}[\rho_{\theta}(x)]\end{equation} for $c\in \{0,1\}$ and task $\Ascr \in \{\Sscr,\Tscr\}$. As mentioned in Sec.~\ref{sec:quantclassification_fixedembedding}, the expected risk \eqref{eq:minexpectedrisk} for task $\Ascr$ is  minimized by the \emph{Hellstrom POVM}, and the resulting minimal expected risk can be obtained in closed form as \cite[Ex. 9.1.7]{wilde2013quantum},
\begin{align}
      \Rscr^{\Ascr}_{\theta}&=\frac{1}{2} - T(p^{\Ascr}_{c}(0)\rho^{\Ascr}_{\theta|0},p^{\Ascr}_{c}(1)\rho^{\Ascr}_{\theta|1})\label{eq:11} 
\end{align}where $p^{\Ascr}_c(c)$ denotes the relevant marginal of the joint distribution $p^{\Ascr}(c,x)$.




\subsection{Measure of Similarity between Source and Target Tasks}\label{sec:similarity}
Of particular interest is the task-based distance $d^{\mathcal{A}}(\theta,\theta_{*}^{\mathcal{A}})$ between any embedding parameter $\theta \in \Theta$ and the embedding parameter $\theta_{*}^{\mathcal{A}}$ that minimizes the expected risk \eqref{eq:minexpectedrisk} for task $\mathcal{A}$, i.e., $\theta_{*}^{\mathcal{A}}=\arg \min_{\theta \in \Theta}  \Rscr^{\mathcal{A}}_{\theta}$, for $\mathcal{A} \in \{\Sscr,\Tscr\}$. This distance measures the sub-optimality of the parameter $\theta$ with respect to the optimal embedding parameter $\theta_{*}^{\mathcal{A}}$ for task $\Ascr$. This is because, by Definition~\ref{def:taskbaseddistance}, a small distance $d^{\Ascr}(\theta,\theta_{*}^{\mathcal{A}})$ implies that the expected risks with embedding parameters $\theta$ and $\theta_{*}^{\mathcal{A}}$ are close. 
Using this idea, and inspired by  \cite[Def. 3]{tripuraneni2020theory}, we introduce the following definition of task dissimilarity. \begin{definition} \label{def:DST} Tasks $\Tscr$ and $\Sscr$ are \emph{$D^{ST}$-dissimilar} if we have the inequality \begin{align}
    d^{\Tscr}(\theta,\theta_{*}^{\Tscr}) \leq d^{\Sscr}(\theta,\theta_{*}^{\Sscr})+D^{ST} \label{eq:DST_def} \end{align} for \emph{all} embedding parameters $\theta \in \Theta$.
\end{definition} 

Hence, the two tasks are $D^{ST}$-dissimilar if the suboptimality of each embedding parameter $\theta$ on the target task  differs from the suboptimality for the source task by no more than a scalar constant $D^{ST}$. Intuitively, a small constant $D^{ST}$ should result in a positive transfer of information from source to target task during pre-training.

The following theorem provides two explicit task dissimilarity measures  satisfying \eqref{eq:DST_def}. To this end, we  define $\TV(p,q)=0.5 \sum_{x \in \Xscr}|p(x)-q(x)|$ as the total variation (TV) distance between discrete distributions $p$ and $q$.
\begin{theorem}\label{lem:relatedness}
For source task, with data distribution $p^{\Sscr}(c,x)$, and target task,  with data distribution $p^{\Tscr}(c,x)$, the following quantities
 \begin{align}
  D^{ST}_{\mathrm{trace}}&= 2 \sup_{\theta \in \Theta} |\Rscr^{\Sscr}_{\theta}-\Rscr^{\Tscr}_{\theta}| \label{eq:DST_trace}, \quad \mbox{and} \\
    D^{ST}_{\mathrm{TV}}&= 2\TV(p^{\Tscr}_c,p^{\Sscr}_c)+2\Ebb_{p^{\Sscr}_c}[\TV(p^{\Tscr}(x|c),p^{\Sscr}(x|c))] \label{eq:DST_TV}
\end{align} satisfy the inequality \eqref{eq:DST_def}. Furthermore, we have  $D^{ST}_{\mathrm{trace}} \leq D^{ST}_{\mathrm{TV}}$.
\end{theorem} 
\begin{proof}
    See Appendix~\ref{app:DST}.
\end{proof}

\section{An Information-Theoretic Bound on The Transfer Excess Risk}\label{sec:infobounds}
In this section, we leverage the measure of dissimilarity $D^{ST}$ between source and target tasks introduced in Section \ref{sec:similarity} to obtain an information-theoretic upper bound on the transfer excess risk \eqref{eq:excess_risk}. Detailed proofs of all results can be found in \cite{jose2022transfer}.

To start, consider the bipartite quantum system comprising of the classical register $X$ reporting the value of input vector $x$ and the quantum register $R_{\theta}$ corresponding to the embedding $\rho_{\theta}(x)$. For a given embedding parameter vector $\theta$, the \emph{classical-quantum state} of the above bipartite system for task $\mathcal{A} \in \{\Sscr,\Tscr\}$ is described by the density matrix
$
    \rho^{\mathcal{A}}_{XR_{\theta}}=\Ebb_{p^{\mathcal{A}}(x)}[\ket{x}\bra{x} \otimes \rho_{\theta}(x)],
$ where $\otimes$ is the Kronecker product, $p^{\Ascr}(x)$ is the relevant marginal of the joint distribution $p^{\Ascr}(c,x)$, and $\{\ket{x}\}$ is an orthonormal basis for the Hilbert space of register $X$, which has dimension equal to the number of possible values for $x$ (see, e.g., \cite{wilde2013quantum}). Then, the \emph{$2$-R\'enyi mutual information (MI)} between the subsystems $X$ and $R_{\theta}$  is defined as  \cite{banchi2021generalization}
\begin{align}\label{eq:renmi}
I^{\Ascr}_2(X;R_{\theta})= 2 \log_2 \Tr \biggl(\sqrt{\sum_x p^{\Ascr}(x)\rho_{\theta}(x)^2}\biggr).
\end{align}

We also define the \textit{Rademacher} complexity of the space  $\Mscr$ of POVM measurements as
\begin{align}
   \mathfrak{R}^{\Ascr}_{\Mscr}\hspace{-0.1cm}=\sup_{\theta \in \Theta} \Ebb_{p^{\Ascr}(c,x)} \Ebb_{p({\sigma})}\hspace{-0.1cm}\biggl[\sup_{ M \in \Mscr}  \sum_{j=1}^{N^{\Ascr}}\frac{\sigma_j \ell_{M,\theta}(c_j,x_j)}{\sqrt{N^{\Ascr}}}  \biggr], \label{eq:radmacher_POVM}
\end{align} and the joint \textit{Rademacher complexity} of the space $\Theta$ of embedding parameters  and of the space $\Mscr$ as
\begin{align}
   \mathfrak{R}^{\Ascr}_{\Theta,\Mscr}\hspace{-0.1cm}=\Ebb_{p^{\Ascr}(c,x)} \Ebb_{p({\sigma})}\hspace{-0.1cm}\biggl[\sup_{\theta\in \Theta, M \in \Mscr}  \sum_{j=1}^{N^{\Ascr}}\frac{\sigma_j \ell_{M,\theta}(c_j,x_j)}{\sqrt{N^{\Ascr}}}  \biggr] \label{eq:cascade_rademacher_main},
\end{align}where the expectation is taken over i.i.d. variables $(c,x) \sim p^{\Ascr}(c,x)$ and over i.i.d. zero-mean and equiprobable Rademacher variables ${\sigma}=(\sigma_1,\hdots,\sigma_{N^{\Ascr}}) \sim p(\sigma)$ with $\sigma_j \in \{+1,-1\}$. The following lemma presents an upper bound on the 
Rademacher complexity  measures \eqref{eq:radmacher_POVM}-\eqref{eq:cascade_rademacher_main}. 
\begin{lemma}\label{lem:rademacher}
Assume that the embedding circuit defines quantum states of the form
$ \rho_{\theta}(x)=U(\theta,x)\vert 0\rangle \langle 0 \vert U(\theta,x)^{\dag}$, where $ U(\theta,x)=\prod_{l=1}^L U_l(\theta_l)S_l(x)$
consists of parameterized unitary gates $U_l(\theta_l)$ as well as encoding gates $S_l(x)$ \cite{schuld2021machine}. The Rademacher complexity measures \eqref{eq:radmacher_POVM} and \eqref{eq:cascade_rademacher_main} can be upper bounded as
 \begin{align}
   \mathfrak{R}^{\Ascr}_{\Mscr}\leq   \mathfrak{R}^{\Ascr}_{\Theta,\Mscr} \leq n c(p^{\Ascr}(x)), \label{eq:ub_jointrademacher}
 \end{align} where $n$ is the dimension of the Hilbert space, and $c(p^{\Ascr}(x))\leq 1$ is a constant that depends on the marginal distribution $p^{\Ascr}(x)$.
\end{lemma}
\begin{proof}
    See Appendix~\ref{app:cascadeRademachercomplexity}.
\end{proof}

\subsection{Upper Bound on Target-Task Excess Risk With No Source-Task Data}
We first consider a baseline scenario when no data from the source task is available. Only data from the target task is used to jointly optimize quantum embedding $\rho_{\theta}(x)$ and measurement $M$. In this case, the excess risk  \eqref{eq:excess_risk_1} for the target task evaluates as $\Delta \Rscr_{\thetah^{\Tscr}}^{\Tscr}$, where $\thetah^{\Tscr}=\arg \min_{\theta \in \Theta} \min_{M \in \Mscr}\widehat{\Rscr}^{\Tscr}_{\theta,M}$. The  following theorem presents an upper bound on the excess risk.
\begin{theorem}\label{thm:PACbound_joint}
The following upper bound on the excess risk $\Delta \Rscr_{\thetah^{\Tscr}}^{\Tscr}$ for the target task holds with probability at least $1-\delta$, for $\delta \in (0,1)$, with respect to the i.i.d. random draws of data set $\Dscr^{\Tscr}$ from the joint distribution $p^{\Tscr}(c,x)$
\begin{align}
 \Delta \Rscr_{\thetah^{\Tscr}}^{\Tscr} \leq    \frac{2(\mathfrak{R}^{\Tscr}_{\Theta,\Mscr}+ \mathfrak{R}^{\Tscr}_{\Mscr}) }{\sqrt{N^{\Tscr}}} + \sqrt{\frac{2}{N^{\Tscr}} \log \frac{2}{\delta}} \label{eq:PACbound_joint},
\end{align}where  $\mathfrak{R}^{\Tscr}_{\Mscr}$ is bounded as
\begin{align}
    \mathfrak{R}^{\Tscr}_{\Mscr} \leq 0.5 \sqrt{ \sup_{\theta \in \Theta}2^{I_2^{\Tscr}(X;R_{\theta})}}, \label{eq:ub_ub}
\end{align}with $I^{\Tscr}_2(X;R_{\theta})$ denoting the 2-R\'enyi MI in \eqref{eq:renmi}. 
\end{theorem}
\begin{proof}
    See Appendix~\ref{app:nosource}.
\end{proof}

The upper bound in \eqref{eq:PACbound_joint} shows that, in the absence of source-task data, the sample complexity scales (at most) proportionally to the sum  $\mathfrak{R}^{\Tscr}_{\Theta,\Mscr}+ \mathfrak{R}^{\Tscr}_{\Mscr}$.


\subsection{Upper Bound on Transfer Excess Risk}

We now present an upper bound on the transfer excess risk \eqref{eq:excess_risk} for the case in which source-task data is available.
\begin{theorem}\label{thm:mainbound}
For any constant $D^{ST}$ satisfying \eqref{eq:DST_def},  the following upper bound on the transfer excess risk holds with probability at least $1-\delta$, for $\delta\in(0,1)$, with respect to the i.i.d random draws of data sets $\mathcal{D}^{\Tscr}$ and $\mathcal{D}^{\Sscr}$ from the respective joint distributions $p^{\mathcal{T}}(c,x)$ and $p^{\mathcal{S}}(c,x)$
\begin{align}
   &\Delta \Rscr^{\Sscr \rightarrow \Tscr} \leq \frac{4 \mathfrak{R}^{\Tscr}_{\Mscr}}{\sqrt{N^{\Tscr}}} +\sqrt{\frac{2}{N^{\Tscr}} \log \frac{3}{\delta}}  + D^{ST} \non \\&+ \frac{2(\mathfrak{R}^{\Sscr}_{\Theta,\Mscr}+ \mathfrak{R}^{\Sscr}_{\Mscr})}{\sqrt{N^{\Sscr}}} + \sqrt{\frac{2}{N^{\Sscr}} \log \frac{3}{\delta}} \label{eq:PACbound}, 
\end{align} where $\mathfrak{R}^{\Ascr}_{\Mscr}$, for $\Ascr \in \{\Sscr,\Tscr\}$, is bounded as in \eqref{eq:ub_ub}.
\end{theorem}
\begin{proof}
    See Appendix~\ref{app:withsource}.
\end{proof}
The bound \eqref{eq:PACbound} illustrates the advantage of transfer learning in reducing the sample complexity for the target task. In fact, if abundant data is available from the source task (i.e., if $N^{\Sscr} \rightarrow \infty$) and if the source and target tasks are sufficiently similar so that $D^{ST}$ is small, the sample complexity of the target task is proportional to  $4\mathfrak{R}^{\Tscr}_{\Mscr}$. By inequality \eqref{eq:ub_jointrademacher}, this is smaller than the scaling $2(\mathfrak{R}^{\Tscr}_{\Theta,\Mscr}+ \mathfrak{R}^{\Tscr}_{\Mscr})$ obtained in Theorem~\ref{thm:PACbound_joint} when no source-task data is available. 
\section{Example and Discussion}\label{sec:example}
In this section, we consider a source task and a target task with equiprobable class label $c \in \{0,1\}$. For each class $c \in \{0,1\}$, we obtain the discrete-valued input $x$ by finely quantizing  a  continuous-valued feature input $\tilde{x} \in \Real$ so that the discrete sum in  \eqref{eq:renmi} can be evaluated via numerical integration \cite{banchi2021generalization}. For the source task, the feature $\tilde{x}$ is Gaussian distributed as   $\Nscr(\tilde{x}|\mu^{\Sscr}_c,\sigma^2)$ with mean $\mu^{\Sscr}_c\in \Real$ and variance $\sigma^2$; while,  for the target task, we have the per-class Gaussian distribution $\Nscr(\tilde{x}|\mu^{\Tscr}_c,\sigma^2)$ with mean $\mu^{\Tscr}_c\in \Real$, generally different from that of source task, and the same variance $\sigma^2$.

The embedding circuit maps the classical input $x$ to the rank-1 density matrix $\rho_{\theta}(x)=\ket{x}\bra{x}$, with the pure quantum state $\ket{x}$  given as
\begin{align}
   \ket{x}=U_{\theta}(x) \ket{0}, \textrm{ with } U_{\theta}(x)=R_X(x)\mathrm{Rot}_{\theta} R_X(x),
\end{align} where $U_{\theta}(x)$ is a unitary matrix  parameterized by the angles  $\theta=(\theta_1,\theta_2,\theta_3) \in [0,2\pi]^3$, which constitutes the embedding PQC  (see, e.g., \cite{schuld2021machine}). The operation of the embedding circuit is  involves the Pauli-X rotation $R_X(x)$ (defined as in \cite[Eq. (3.45)]{schuld2021machine})
and the general rotation $\Rot_{\theta}$ defined as in \cite[Eq. (3.48)]{schuld2021machine}.

    \begin{figure}
        \centering
       \includegraphics[scale=0.39,trim=2.3in 0.7in 2.6in 0.8in ,clip=true]{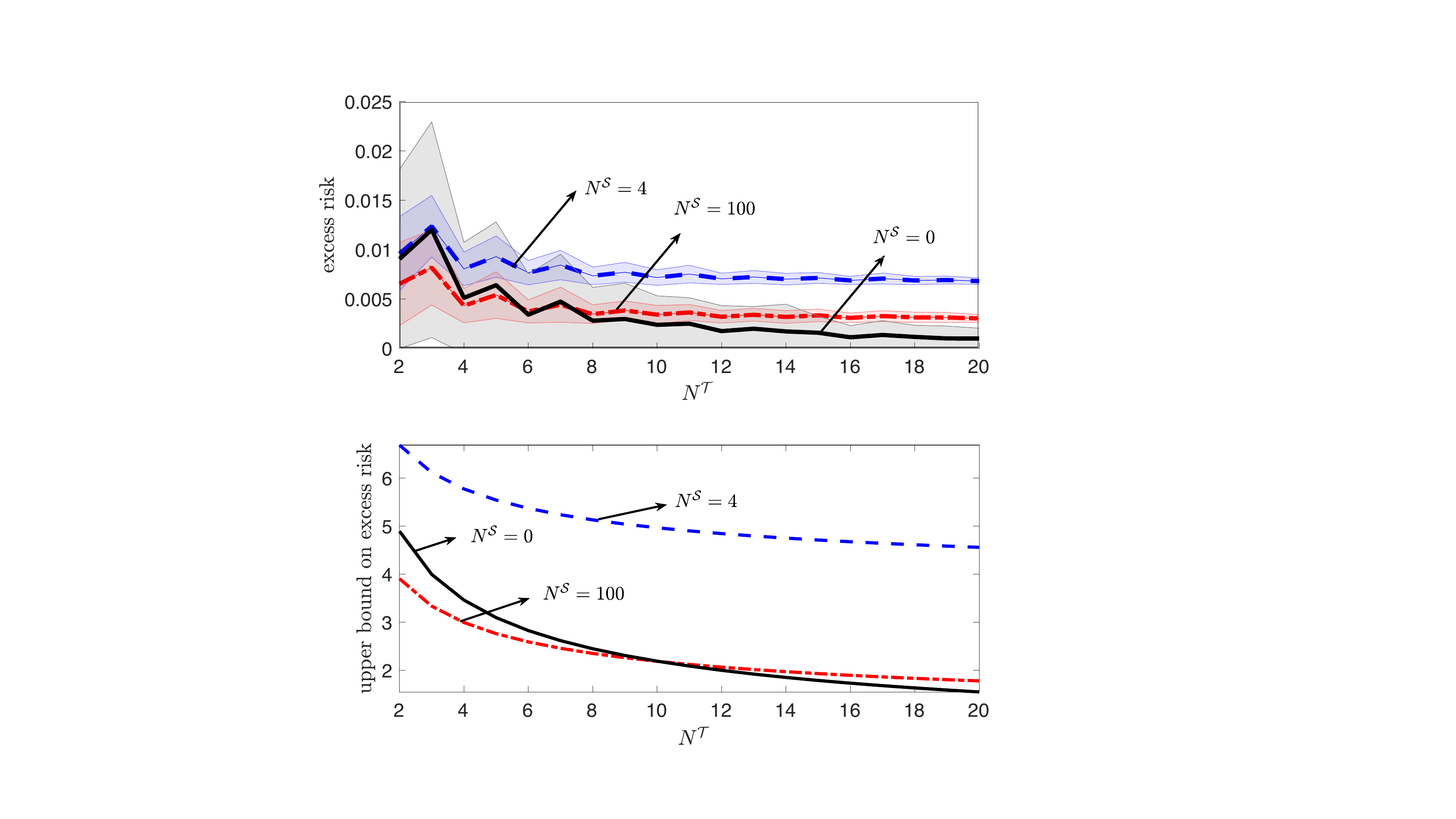}
        \caption{Transfer excess risk  in \eqref{eq:excess_risk} (top) and upper bound \eqref{eq:PACbound} (bottom) as a function of the target-task samples $N^{\Tscr}$ for varying values of source-task samples. $N^{\Sscr}=0$ corresponds to the excess risk $\Delta \Rscr^{\Tscr}_{\thetah^{\Tscr}}$ and the upper bound \eqref{eq:PACbound_joint}. }
      \label{fig:samplesperformance}
      \vspace{-0.3cm}
    \end{figure}

In Figure~\ref{fig:samplesperformance}, we plot the transfer excess risk $\Delta \Rscr^{\Sscr \rightarrow \Tscr}$ (top figure), along with the corresponding upper bound derived in \eqref{eq:PACbound} (bottom figure) as a function of the number of target task samples $N^{\Tscr}$ for varying values of source task samples $N^{\Sscr}$. Note that $N^{\Sscr}=0$ corresponds to the excess risk $\Delta \Rscr^{\Tscr}_{\thetah^{\Tscr}}$ and the upper bound \eqref{eq:PACbound_joint}. Other parameters are set as $\delta=0.5$, $\sigma^2=0.11$, $\mu^{\Sscr}_0=1$, $\mu^{\Sscr}_1=-1$, $\mu^{\Tscr}_0=1.5$,  and $\mu^{\Tscr}_1=-0.5$.  The transfer excess risk $\Delta \Rscr^{\Sscr \rightarrow \Tscr}$  is a random variable, which is evaluated by drawing multiple pairs of data sets $(\mathcal{D}^{\Sscr},\mathcal{D}^{\Tscr})$ from their respective joint distributions $p^{\Sscr}(c,x)$ and $p^{\Tscr}(c,x)$. The thick lines in the top figure correspond to the median of the resulting empirical distribution, while the shaded areas represent its spread.

The figure shows that the upper bound \eqref{eq:PACbound}, while numerically loose (as is common for related information-theoretic bounds in classical machine learning \cite{wu2020information,jose2021information}),  predicts well the regime where transfer learning is advantageous. Comparing the case when no source task is available (i.e., $N^{\Sscr}=0$) to when abundant source-task data is available for transfer learning (i.e., $N^{\Sscr}=100$), Fig~\ref{fig:samplesperformance} shows that transfer learning can achieve a smaller excess risk when limited data are available from the target task ($N^{\Tscr}<8$). This advantage vanishes when target task-data become increasingly available, in which case the contribution of the task dissimilarity measure $D^{ST}$ to transfer excess risk outweighs the other terms in \eqref{eq:PACbound_joint}. This also explains the non-vanishing behaviour of transfer excess risk in Fig ~\ref{fig:samplesperformance} in the limit as $N^{\Sscr}$ and  $N^{\Tscr}\rightarrow \infty$ grow large.




\begin{figure}
    \centering
    \includegraphics[scale=0.3,trim=1.6in 1in 2.6in 0.9in ,clip=true]{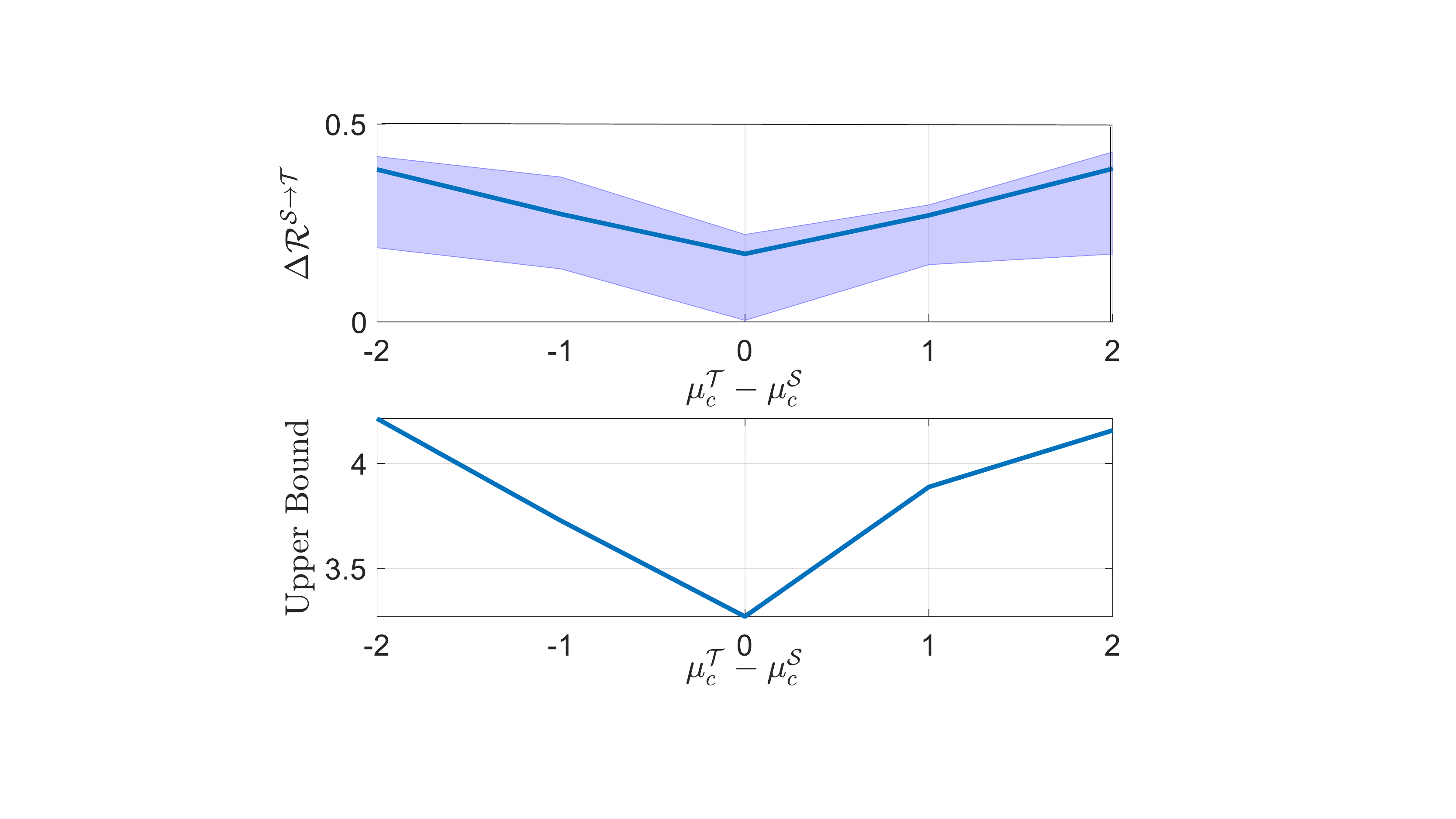}
    \caption{Transfer excess risk $\Delta{\Rscr}^{\Sscr \rightarrow \Tscr}$ in (\ref{eq:excess_risk}) (top) and upper bound \eqref{eq:PACbound} (bottom) as a function of the deviation, $\mu^{\Tscr}_c-\mu^{\Sscr}_c$, of the means of source and target tasks.} 
    \label{fig:domaindivergence}
    \vspace{-0.6cm}
\end{figure}

The impact of the dissimilarity between the two tasks is further elaborated on in Figure~\ref{fig:domaindivergence}, which illustrates  the transfer excess risk $\Delta \Rscr^{\Sscr \rightarrow \Tscr}$ (top) and the corresponding upper bound in \eqref{eq:PACbound} (bottom) as a function of the difference $\mu^{\Tscr}_c-\mu^{\Sscr}_c$ between the means of the input data under the target and source tasks for both classes $c \in \{0,1\}$. We fix $\mu^{\Sscr}_0=1$ and $\mu^{\Sscr}_1=-2$, $N^{\Tscr}=4$,  $N^{\Sscr}=10$, $\sigma^2=1$, and $\delta=0.9$. As can be seen, when the difference between the means is zero, \ie, when the source and target tasks coincide, the transfer excess risk, and the corresponding upper bound, are minimized. Conversely, when the mean of the target task deviates from that of the source task, the transfer excess risk increases, as correctly predicted by the upper bound. 
\bibliography{ref}
\bibliographystyle{IEEEtran} 
\appendices
\section{Proof of Theorem~\ref{thm:PACbound_joint}}\label{app:nosource}
Throughout the appendix, we make use of the following convention. For fixed $\theta \in \Theta$, and task $\Ascr \in \{\Sscr,\Tscr\}$,  we define $ M^{\Ascr}_{\theta}=\arg \min_{M \in \Mscr} \Rscr^{\Ascr}_{\theta,M}$ as the optimal measurement minimizing the expected risk, and $ \widehat{M}^{\Ascr}_{\theta}=\arg \min_{M \in \Mscr} \widehat{\Rscr}^{\Ascr}_{\theta,M}$ as the optimal measurement minimizing the empirical training loss.

To obtain the required upper bound in \eqref{eq:PACbound_joint}, we start by decomposing the excess risk $\Delta \Rscr_{\thetah^{\Tscr}}^{\Tscr}$ as follows:
\begin{align}
    \Delta \Rscr_{\thetah^{\Tscr}}^{\Tscr} &= \Rscr^{\Tscr}_{\thetah^{\Tscr},\widehat{M}_{\thetah^{\Tscr}}^{\Tscr}} - \widehat{\Rscr}^{\Tscr}_{\thetah^{\Tscr},\widehat{M}_{\thetah^{\Tscr}}^{\Tscr}}+\widehat{\Rscr}^{\Tscr}_{\thetah^{\Tscr},\widehat{M}_{\thetah^{\Tscr}}^{\Tscr}}-\widehat{\Rscr}^{\Tscr}_{\theta^{\Tscr}_{*},{M}_{\theta^{\Tscr}_{*}}^{\Tscr}} \non \\
    &+ \widehat{\Rscr}^{\Tscr}_{\theta^{\Tscr}_{*},{M}^{\Tscr}_{\theta^{\Tscr}_{*}}}-{\Rscr}^{\Tscr}_{\theta^{\Tscr}_{*},{M}^{\Tscr}_{\theta^{\Tscr}_{*}}} \non \\
    & \stackrel{(a)}{\leq}\Rscr^{\Tscr}_{\thetah^{\Tscr},\widehat{M}^{\Tscr}_{\thetah^{\Tscr}}} - \widehat{\Rscr}^{\Tscr}_{\thetah^{\Tscr},\widehat{M}^{\Tscr}_{\thetah^{\Tscr}}}+ \widehat{\Rscr}^{\Tscr}_{\theta^{\Tscr}_{*},{M}^{\Tscr}_{\theta^{\Tscr}_{*}}}-{\Rscr}^{\Tscr}_{\theta^{\Tscr}_{*},{M}^{\Tscr}_{\theta^{\Tscr}_{*}}}\non \\
    &\leq \sup_{\theta \in \Theta} \Gscr^{\Tscr}_{\theta}(\Dscr^{\Tscr}) + \Gscr^{\Tscr}_{\theta^{\Tscr}_{*}}(\Dscr^{\Tscr}), \label{eq:new_1}
\end{align} where \begin{align}
    \Gscr^{\Ascr}_{\theta}(\Dscr^{\Ascr})=\sup_{M \in \Mscr}|\Rscr^{\Ascr}_{\theta,M}-\widehat{\Rscr}^{\Ascr}_{\theta,M}| \label{eq:generror}\end{align} is the generalization error of task $\Ascr$ with fixed embedding parameter $\theta \in \Theta$. The inequality in $(a)$ follows by the inequality $\widehat{\Rscr}^{\Tscr}_{\thetah^{\Tscr},\widehat{M}^{\Tscr}_{\thetah^{\Tscr}}}\leq \widehat{\Rscr}^{\Tscr}_{\theta^{\Tscr}_{*},{M}^{\Tscr}_{\theta^{\Tscr}_{*}}}$.

We now separately upper bound each of the terms in \eqref{eq:new_1} with high probability with respect to random draws of the dataset $\Dscr^{\Tscr}$, and then combine the resulting bounds via union bound. To this end, note
 that the loss function  $\ell_{M,\theta}(c,x)$ in \eqref{eq:loss} is $[0,1]$-bounded,  and hence the classical Radmacher complexity-based generalization bound  gives that 
with probability at least $1-\delta'$ over the random draw of data $\Dscr^{\Ascr}\sim p^{\Ascr}(c,x)^{\otimes N^{\Ascr}}$, for $\Ascr \in \{\Sscr,\Tscr\}$, the following inequality holds \cite{shalev2014understanding}
    \begin{align}
      \Gscr^{\Ascr}_{\theta}(\Dscr^{\Ascr})
      &\leq 2 \frac{\mathfrak{R}^{\Ascr}_{\theta,\Mscr} }{\sqrt{N^{\Ascr}}} + \sqrt{\frac{1}{2N^{\Ascr}} \log \Bigl(\frac{1}{\delta'}\Bigr)} \label{eq:int_2}\\
      &\leq 2  \frac{\mathfrak{R}^{\Ascr}_{\Mscr} }{\sqrt{N^{\Ascr}}}+ \sqrt{\frac{1}{2N^{\Ascr}} \log \Bigl(\frac{1}{\delta'}\Bigr)}
      \label{eq:pointer3},
    \end{align} where
    \begin{align}
   \mathfrak{R}^{\Ascr}_{\theta,\Mscr}\hspace{-0.1cm}= \Ebb_{p^{\Ascr}(c,x)} \Ebb_{p({\sigma})}\hspace{-0.1cm}\biggl[\sup_{ M \in \Mscr}  \sum_{j=1}^{N^{\Ascr}}\frac{\sigma_j \ell_{M,\theta}(c_j,x_j)}{\sqrt{N^{\Ascr}}}  \biggr], 
\end{align} and $\mathfrak{R}^{\Ascr}_{\Mscr} =\sup_{\theta \in \Theta}\mathfrak{R}^{\Ascr}_{\theta,\Mscr}$ is defined in \eqref{eq:radmacher_POVM}. Furthermore, it  follows from \cite[Thm. 2]{banchi2021generalization} that for binary classification with fixed $\theta \in \Theta$, we have
    \begin{align}
      \mathfrak{R}^{\Ascr}_{\theta,\Mscr} \leq  \frac{1}{2} \sqrt{ 2^{I_2^{\Ascr}(X;R_{\theta})}}  \label{eq:ub_1},
    \end{align} where $I_2^{\Ascr}(X;R_{\theta})$ is the 2-Renyi MI between subsystems $X$ and $R_{\theta}$ of task $\Ascr$ for fixed $\theta \in \Theta$. Using \eqref{eq:ub_1} and \eqref{eq:pointer3} in \eqref{eq:new_1} gives a high probability bound on $\Gscr^{\Tscr}_{\theta^{\Tscr}_{*}}(\Dscr^{\Tscr})$ which holds with probability at least $1-\delta'$ for task $\Tscr$.
    
The term $\sup_{\theta \in \Theta} \Gscr^{\Ascr}_{\theta}(\Dscr^{\Ascr})$ can be bounded as in \eqref{eq:int_2}. This gives that with probability at least $1-\delta'$ over the random draw of data $\Dscr^{\Ascr} \sim p^{\Ascr}(c,x)^{\otimes N^{\Ascr}}$, we have
    \begin{align}
     \sup_{\theta \in \Theta} \Gscr^{\Ascr}_{\theta}(\Dscr^{\Ascr}) \leq 2 \frac{\mathfrak{R}^{\Ascr}_{\Theta,\Mscr} }{\sqrt{N^{\Ascr}}} +   \sqrt{\frac{1}{2N^{\Ascr}} \log \Bigl(\frac{1}{\delta'}\Bigr)} \label{eq:rad_ub}, 
     \end{align} where $\mathfrak{R}^{\Ascr}_{\Theta,\Mscr}$ is defined in \eqref{eq:cascade_rademacher_main}. Now, combining the upper bounds obtained on each of the terms in \eqref{eq:new_1} via union bound with the choice of $\delta'=\delta/2$ yields the  bound of \eqref{eq:PACbound_joint}. 
\section{Proof of Theorem~\ref{thm:mainbound}}\label{app:withsource}
To derive the upper bound in Theorem~\ref{thm:mainbound}, we make use of the following auxiliary lemma, the proof of which can be found in Appendix~\ref{app:proof_auxiliary}. 
\begin{lemma}\label{lem:auxiliarylemma}
Let $|\Theta|>1$. For any scalar constant $D^{ST}$ satisfying \eqref{eq:DST_def}, the following upper bound on the transfer excess risk holds
\begin{align}
    \Delta\Rscr^{\Sscr \rightarrow \Tscr} &\leq 2 \Gscr_{\thetah}^{\Tscr}(\Dscr^{\Tscr})+ d^{\Sscr}(\thetah, \theta_{*}^{\Tscr}) +D^{ST} \label{eq:bound_1}\\
    &\leq 2 \Gscr_{\thetah}^{\Tscr}(\Dscr^{\Tscr})+ \sup_{\theta \in \Theta}\Gscr_{\theta}^{\Sscr}(\Dscr^{\Sscr})+\Gscr_{\theta^{\Sscr}_{*}}^{\Sscr}(\Dscr^{\Sscr}) +D^{ST}, \label{eq:bound_2}
\end{align} where $\Gscr_{\theta}^{\Ascr}(\Dscr^{\Ascr})$ is defined as in \eqref{eq:generror}.
Furthermore, if $\Theta=\{\theta\}$ is a singleton set, the bound on  transfer excess risk can be tightened as
\begin{align}
  \Delta\Rscr^{\Sscr \rightarrow \Tscr} &\leq 2 \Gscr_{\theta}^{\Tscr}(\Dscr^{\Tscr})  \label{eq:bound_3}.
\end{align}
\end{lemma}

We are now ready to state the proof of Theorem~\ref{thm:mainbound}. For the case when $|\Theta|>1$, the  upper bound  \eqref{eq:PACbound} is obtained by 
upper bounding each of the first three terms in  \eqref{eq:bound_2} with high probability with respect to the training data sets, and then combining the resulting bounds via the union bound.

To this end, note that, for fixed pre-trained embedding parameter $\thetah \in \Theta$, the term $\Gscr^{\Tscr}_{\thetah}(\Dscr^{\Tscr})$ in \eqref{eq:bound_2} can be upper bounded with high probability over random draws of data sets $\Dscr^{\Tscr}$ as in \eqref{eq:pointer3} by using \eqref{eq:ub_1}. In a similar way, for fixed $\theta^{\Sscr}_{*} \in \Theta$, the term $\Gscr^{\Sscr}_{\theta^{\Sscr}_{*}}(\Dscr^{\Sscr})$ can be upper bounded with high probability over random draws of data sets $\Dscr^{\Sscr}$. To upper bound  the term $\sup_{\theta \in \Theta}\Gscr_{\theta}^{\Sscr}(\Dscr^{\Sscr})$ which holds with high probability with respect to random draws of data sets $\Dscr^{\Sscr}$, we use \eqref{eq:rad_ub}. Finally, combining the upper bounds obtained on each of the first three terms in \eqref{eq:bound_2} via  union bound with the choice of $\delta'=\delta/3$ yields the required bound in \eqref{eq:PACbound}.
 

Finally, for the case when $|\Theta|=1$, we upper bound \eqref{eq:bound_3} as in \eqref{eq:int_2} to get the inequality \eqref{eq:bound_3}.
    \section{Proof of Lemma~\ref{lem:auxiliarylemma}}\label{app:proof_auxiliary}

To obtain an upper bound on transfer excess risk,  we start by
decomposing it as
\begin{align}
 \Delta\Rscr^{\Sscr \rightarrow \Tscr}&= \underbrace{\Rscr^{\Tscr}_{\thetah,\widehat{M}^{\Tscr}_{\thetah}, }-\Rscr^{\Tscr}_{\thetah,M^{\Tscr}_{\thetah}}}_{\Escr_{\Tscr}(\thetah)}+ \Rscr^{\Tscr}_{\thetah,M^{\Tscr}_{\thetah}, }-\Rscr^{\Tscr}_{\theta_{*}^{\Tscr},M^{\Tscr}_{ \theta_{*}^{\Tscr}}}  \\
& \stackrel{(a)}{\leq} 2\Gscr^{\Tscr}_{\thetah}(\Dscr^{\Tscr})+ \Rscr^{\Tscr}_{\thetah,M^{\Tscr}_{\thetah}}-\Rscr^{\Tscr}_{\theta_{*}^{\Tscr},M^{\Tscr}_{\theta_{*}^{\Tscr}}} \\
&=2\Gscr^{\Tscr}_{\thetah}(\Dscr^{\Tscr})+ d^{\Tscr}(\thetah,\theta^{\Tscr}_{*}),
\label{eq:pointer1}
    \end{align} where the upper bound in $(a)$ follows by a canonical decomposition of $\Escr_{\Tscr}(\thetah)$ as follows:
    \begin{align}
    \Escr_{\Tscr}(\thetah)&=    \Rscr^{\Tscr}_{\thetah,\widehat{M}^{\Tscr}_{\thetah}}-\widehat{\Rscr}^{\Tscr}_{\thetah,\widehat{M}^{\Tscr}_{\thetah}}+\underbrace{\widehat{\Rscr}^{\Tscr}_{\thetah,\widehat{M}^{\Tscr}_{\thetah}}-\widehat{\Rscr}^{\Tscr}_{\thetah,M^{\Tscr}_{\thetah}}}_{\leq 0}\non \\& \qquad +\widehat{\Rscr}^{\Tscr}_{\thetah,M^{\Tscr}_{\thetah}}-\Rscr^{\Tscr}_{\thetah,M^{\Tscr}_{\thetah}} \nonumber \\
    & \leq \Rscr^{\Tscr}_{\thetah,\widehat{M}^{\Tscr}_{\thetah}}-\widehat{\Rscr}^{\Tscr}_{\thetah,\widehat{M}^{\Tscr}_{\thetah}}+\widehat{\Rscr}^{\Tscr}_{\thetah,M^{\Tscr}_{\thetah}}-\Rscr^{\Tscr}_{\thetah,M^{\Tscr}_{\thetah}} \leq 2 \Gscr^{\Tscr}_{\thetah}(\Dscr^{\Tscr}).
    \end{align}
    Note here that if $|\Theta|=1$, then $d^{\Tscr}(\thetah,\theta^{\Tscr}_{*})=0$ in \eqref{eq:pointer1}, and we get the upper bound in \eqref{eq:bound_3}. For $|\Theta|>1$, 
   the bound in \eqref{eq:bound_1} follows from  \eqref{eq:pointer1} by using the inequality \eqref{eq:DST_def}.
    
We now upper bound the  source task-distance $d^{\Sscr}(\thetah,\theta_{*}^{\Sscr})$ in \eqref{eq:bound_1}. Towards this, we note that the following sequence of inequalities hold for any measurement $M \in \Mscr$, \begin{align}
    d^{\Sscr}(\thetah,\theta_{*}^{\Sscr})&=\Rscr^{\Sscr}_{\thetah,M^{\Sscr}_{\thetah}}-\Rscr^{\Sscr}_{\theta_{*}^{\Sscr},M^{\Sscr}_{\theta_{*}^{\Sscr}}}\\
    & \stackrel{(a)}{\leq}\Rscr^{\Sscr}_{\thetah,M}-\Rscr^{\Sscr}_{\theta_{*}^{\Sscr},M^{\Sscr}_{\theta_{*}^{\Sscr}}}\\
    &=\Rscr^{\Sscr}_{\thetah,M}-\widehat{\Rscr}^{\Sscr}_{\thetah,M}+\widehat{\Rscr}^{\Sscr}_{\thetah,M}-\widehat{\Rscr}^{\Sscr}_{\theta_{*}^{\Sscr},M^{\Sscr}_{\theta_{*}^{\Sscr}}}\non \\&+\widehat{\Rscr}^{\Sscr}_{\theta_{*}^{\Sscr},M^{\Sscr}_{\theta_{*}^{\Sscr}}}-\Rscr^{\Sscr}_{\theta_{*}^{\Sscr},M^{\Sscr}_{\theta_{*}^{\Sscr}}} \label{eq:pointer_2}
\end{align}where the inequality in $(a)$ follows since $\Rscr^{\Sscr}_{\thetah,M^{\Sscr}_{\thetah}} =\min_{M \in \Mscr} \Rscr^{\Sscr}_{\thetah,M} \leq \Rscr^{\Sscr}_{\thetah,M} $ for all $M \in \Mscr$. 
In particular, choosing $M=\widehat{M}^{\Sscr}_{\thetah}$ in \eqref{eq:pointer_2}, yields the second difference of \eqref{eq:pointer_2} bounded as
    $
        \widehat{R}^{\Sscr}_{\thetah, \widehat{M}^{\Sscr}_{\thetah}}-\widehat{R}^{\Sscr}_{\theta^{\Sscr}_{*},M^{\Sscr}_{\theta^{\Sscr}_{*}}}\leq 0,
   $
   whereby we have that 
   \begin{align}
    d^{\Sscr}(\thetah,\theta_{*}^{\Sscr})& \leq \Rscr^{\Sscr}_{\thetah,\widehat{M}^{\Sscr}_{\thetah}}-\widehat{\Rscr}^{\Sscr}_{\thetah,\widehat{M}^{\Sscr}_{\thetah}}+\widehat{\Rscr}^{\Sscr}_{\theta_{*}^{\Sscr},M^{\Sscr}_{\theta_{*}^{\Sscr}}}-\Rscr^{\Sscr}_{\theta_{*}^{\Sscr},M^{\Sscr}_{\theta_{*}^{\Sscr}}} \non \\
    & \leq \sup_{\theta \in \Theta}\Gscr^{\Sscr}_{\theta}(\Dscr^{\Sscr}) + \Gscr^{\Sscr}_{\theta^{\Sscr}_{*}}(\Dscr^{\Sscr}). \label{eq:pointer5}
   \end{align} Using \eqref{eq:pointer5} in \eqref{eq:bound_1} yields \eqref{eq:bound_2}.
\section{Proof of Theorem~3.1}\label{app:DST}

To derive the dissimilarity measure $D^{ST}$, we start by considering the difference 
\begin{align}
    &d^{\Tscr}(\theta,\theta^{\Tscr}_{*})-d^{\Sscr}(\theta,\theta^{\Sscr}_{*})\\
    &\stackrel{(a)}{=} T^{\Tscr}(\theta_{*}^{\Tscr})-T^{\Tscr}(\theta) -T^{\Sscr}(\theta_{*}^{\Sscr})+T^{\Sscr}(\theta)\\
    & \stackrel{(b)}{\leq} T^{\Tscr}(\theta_{*}^{\Tscr})-T^{\Sscr}(\theta_{*}^{\Tscr})-T^{\Tscr}(\theta) +T^{\Sscr}(\theta)\\
    & \leq 2 \sup_{\theta \in \Theta}|T^{\Tscr}(\theta)-T^{\Sscr}(\theta)|:=D^{ST}_{\mathrm{trace}} \label{eq:DST_appendix},
\end{align}where the equality in $(a)$ follows from Lemma A.1 with $T^{\mathcal{A}}(\theta)=T(p^{\Ascr}_c(0)\rho_{\theta|0}^{\Ascr},p^{\Ascr}_c(1)\rho_{\theta|1}^{\Ascr})$ denoting the inter-class trace distance; and the inequality in $(b)$ follows since $T^{\Sscr}(\theta^{\Sscr}_{*}) \geq T^{\Sscr}(\theta)$ for all $\theta \in \Theta$ (Lemma A.1) which in turn implies that $T^{\Sscr}(\theta^{\Sscr}_{*}) \geq T^{\Sscr}(\theta^{\Tscr}_{*})$.

To derive the upper bound of \eqref{eq:DST_TV} on $D^{ST}_{\mathrm{trace}}$, we start by noting that for any $\theta \in \Theta$, we have that
\begin{align}
    &T^{\Tscr}(\theta)-T^{\Sscr}(\theta)\non \\&=\frac{1}{2}\lVert p_c^{\Tscr}(0)\rho^{\Tscr}_{\theta|0}-p_c^{\Tscr}(1)\rho^{\Tscr}_{\theta|1}\rVert_1-\frac{1}{2} \lVert p_c^{\Sscr}(0)\rho^{\Sscr}_{\theta|0}-p_c^{\Sscr}(1)\rho^{\Sscr}_{\theta|1}\rVert_1  \non\\
    & \stackrel{(a)}{=}\frac{1}{2} \max_U \Bigl| \Tr\Bigl(U( p_c^{\Tscr}(0)\rho^{\Tscr}_{\theta|0}-p_c^{\Tscr}(1)\rho^{\Tscr}_{\theta|1})\Bigr)\Bigr|\non \\&-\frac{1}{2}\max_U \Bigl|\Tr\Bigl(U( p_c^{\Sscr}(0)\rho^{\Sscr}_{\theta|0}-p_c^{\Sscr}(1)\rho^{\Sscr}_{\theta|1})\Bigr)\Bigr| \\
    &=\frac{1}{2}  \Bigl| \Tr\Bigl(U^{*}( p_c^{\Tscr}(0)\rho^{\Tscr}_{\theta|0}-p_c^{\Tscr}(1)\rho^{\Tscr}_{\theta|1})\Bigr)\Bigr|\non \\&-\frac{1}{2}\max_U \Bigl|\Tr\Bigl(U( p_c^{\Sscr}(0)\rho^{\Sscr}_{\theta|0}-p_c^{\Sscr}(1)\rho^{\Sscr}_{\theta|1})\Bigr)\Bigr| \\
    &\leq \frac{1}{2}  \Bigl| \Tr\Bigl(U^{*}( p_c^{\Tscr}(0)\rho^{\Tscr}_{\theta|0}-p_c^{\Tscr}(1)\rho^{\Tscr}_{\theta|1})\Bigr)\Bigr|\non \\&-\frac{1}{2} \ \Bigl|\Tr\Bigl(U^{*}( p_c^{\Sscr}(0)\rho^{\Sscr}_{\theta|0}-p_c^{\Sscr}(1)\rho^{\Sscr}_{\theta|1})\Bigr)\Bigr| \non
    \\
    & \stackrel{(b)}{\leq} \frac{1}{2} \biggl|   \Tr\Bigl(U^{*}( p_c^{\Tscr}(0)\rho^{\Tscr}_{\theta|0}-p_c^{\Tscr}(1)\rho^{\Tscr}_{\theta|1} \non \\&- p_c^{\Sscr}(0)\rho^{\Sscr}_{\theta|0}+p_c^{\Sscr}(1)\rho^{\Sscr}_{\theta|1})\Bigr)\Biggr|\non \\
    &\leq \frac{1}{2} \Bigl|   \Tr\Bigl(U^{*}( p_c^{\Tscr}(0)\rho^{\Tscr}_{\theta|0}-p_c^{\Sscr}(0)\rho^{\Sscr}_{\theta|0})\Bigr)\Bigr| \non \\&+\frac{1}{2}\Bigl| \Tr\Bigl(U^{*}( p_c^{\Sscr}(1)\rho^{\Sscr}_{\theta|1}-p_c^{\Tscr}(1)\rho^{\Tscr}_{\theta|1})\Bigr)\Bigr| \non \\&\leq  T(p_c^{\Tscr}(0)\rho^{\Tscr}_{\theta|0},p_c^{\Sscr}(0)\rho^{\Sscr}_{\theta|0}) + T(p_c^{\Sscr}(1)\rho^{\Sscr}_{\theta|1},p_c^{\Tscr}(1)\rho^{\Tscr}_{\theta|1}) \label{eq:intermediate_1} 
\end{align} where the equality in $(a)$ follows from the variational representation of the trace norm, i.e,
$\lVert M \rVert_1 = \max_U |\Tr(MU)|$
with the maximization done over the space of all unitary operators. The inequality in $(b)$ follows since for any $x,y \in \mathbb{R}$, $|x|-|y|\leq |x-y|$. The final inequality again follows from the variational representation of trace norm  and the definition of trace distance in \eqref{eq:trdist}.

In a similar way, one can verify that the inequality 
$T^{\Sscr}(\theta)-T^{\Tscr}(\theta) \leq \sum_{c 
\in \{0,1\}} T(p_c^{\Tscr}(c)\rho^{\Tscr}_{\theta,c},p_c^{\Sscr}(c)\rho^{\Sscr}_{\theta,c})$ holds. Consequently, we  have that
\begin{align}
    |T^{\Tscr}(\theta)-T^{\Sscr}(\theta)| \leq \sum_{c 
\in \{0,1\}} T(p_c^{\Tscr}(c)\rho^{\Tscr}_{\theta,c},p_c^{\Sscr}(c)\rho^{\Sscr}_{\theta,c}).
\end{align}

Further, the trace distance $T(p_c^{\Tscr}(c)\rho^{\Tscr}_{\theta|c},p_c^{\Sscr}(c)\rho^{\Sscr}_{\theta|c})$ can be upper bounded as
\begin{align}
&2T(p_c^{\Tscr}(c)\rho^{\Tscr}_{\theta|c},p_c^{\Sscr}(c)\rho^{\Sscr}_{\theta|c})\non \\&=\max_U \Bigl|\Tr\Bigl(U(p_c^{\Tscr}(c)\rho^{\Tscr}_{\theta|c}-p_c^{\Sscr}(c)\rho^{\Sscr}_{\theta|c} )\Bigr) \Bigr|\non \\
&=\max_U \Bigl|\sum_x \Tr\Bigl(U\Bigl(p^{\Tscr}(c,x)\rho_{\theta}(x)-p^{\Sscr}(c,x)\rho_{\theta}(x) \Bigr)\Bigr) \Bigr|\non \\
&=\max_U \Bigl|\sum_x p^{\Tscr}(x|c) \Tr\Bigl(U\Bigl(p_c^{\Tscr}(c)\rho_{\theta}(x)-p_c^{\Sscr}(c)\rho_{\theta}(x)\Bigr)\Bigr) \non \\&+\sum_x (p^{\Tscr}(x|c)-p^{\Sscr}(x|c)) p^{\Sscr}_c(c)\Tr(U\rho_{\theta}(x)) \Bigr| \non \\
&\leq \max_U \Bigl|(p^{\Tscr}_c(c)-p^{\Sscr}_c(c))\Tr\Bigl(U \rho^{\Tscr}_{\theta|c}\Bigr) \Bigr| \non \\
&+\sum_x |(p^{\Tscr}(x|c)-p^{\Sscr}(x|c))|p^{\Sscr}_c(c) \max_U |\Tr(U\rho_{\theta}(x))|\non \\
&\leq |p^{\Tscr}_c(c)-p^{\Sscr}_c(c)|\lVert \rho_{\theta,c}^{\Tscr}\rVert_1\non \\&+p^{\Sscr}_c(c)\sum_x |(p^{\Tscr}(x|c)-p^{\Sscr}(x|c))| \lVert \rho_{\theta}(x)\rVert_1 \non \\
&=|p^{\Tscr}_c(c)-p^{\Sscr}_c(c)|+p^{\Sscr}_c(c)\sum_x |(p^{\Tscr}(x|c)-p^{\Sscr}(x|c))|  \non \\
&=|p^{\Tscr}_c(c)-p^{\Sscr}_c(c)|+2p^{\Sscr}_c(c)\TV(p^{\Tscr}(x|c),p^{\Sscr}(x|c))).
\end{align}
We thus have that
\begin{align}
  D^{ST}_{\mathrm{trace}}&=  2 \sup_{\theta \in \Theta}|T^{\Tscr}(\theta)-T^{\Sscr}(\theta)| \non \\
  &\leq 2\biggl(\TV(p^{\Tscr}_c,p^{\Sscr}_c)+\Ebb_{p^{\Sscr}_c}[\TV(p^{\Tscr}(x|c),p^{\Sscr}(x|c)))] \biggr)\non \\
  &=D^{ST}_{\mathrm{TV}}.
\end{align}

\section{Proof of Lemma~\ref{lem:rademacher}}\label{app:cascadeRademachercomplexity}
In this section, we derive an upper bound on the Rademacher complexity $\mathfrak{R}^{\Ascr}_{\Theta,\Mscr}$ of the joint space of parameterized PQCs and measurement operators that depends on $n$, the dimension of the Hilbert space, as well as the marginal distribution $p^{\Ascr}(x)$, for the  following class of variational quantum circuits.
\subsection{General Ansatz}
Assume that the PQC is implemented via the unitary gate,
\begin{align}
    U(\theta,x)=\prod_{l=1}^L U_l(\theta_l)S_l(x), \label{eq:VQC}
\end{align}which consists of $L$ layers of alternating parameterized unitary gates $U_l(\theta_l)$ and data embedding gates $S_l(x)$ \cite{sweke2020stochastic}. Each parameterized gate $U_l(\theta_l)$ acts on $k$ qubits. The pure state density matrix $\rho_{\theta}(x)$ is then obtained via the operation of the unitary $U(\theta,x)$  on an initial quantum state $\vert 0\rangle$ as $\rho_{\theta}(x)=U(\theta,x)\vert 0\rangle \langle 0 \vert U^{\dag}(\theta,x)$. Note that the PQC in \eqref{eq:VQC} accounts for one-time data encoding strategy (with $S_l(x)=S(x)$), as well as repeated encoding strategies, and thus it describes a large class of PQCs used in quantum machine learning \cite{schuld2021machine}.

To bound the joint Rademacher complexity, we equivalently write the loss function as $\ell_{M,\theta}(c_j,x_j)= \Tr(\delta_{c_j}(0)M_1 \rho_{\theta}(x_j)+\delta_{c_j}(1)M_0 \rho_{\theta}(x_j))$, where $\delta_c(a)$ is the indicator function which takes value $1$ when $c=a$ and is zero otherwise. Substituting $M_0=I-M_1$, we get that $\ell_{M,\theta}(c_j,x_j)= \delta_{c_j}(1)+\Tr(M_1(\delta_{c_j}(0)-\delta_{c_j}(1)) \rho_{\theta}(x_j))$. Thus,
\begin{align}
    &\mathfrak{R}^{\Ascr}_{\Theta,\Mscr}
     =\Ebb_{p^{\Ascr}(c,x)} \Ebb_{p({\sigma})}\biggl[\sup_{\theta\in \Theta, 0\leq M \leq I} \frac{1}{\sqrt{N^{\Ascr}}} \sum_{j=1}^{N^{\Ascr}} \sigma_j \Bigl(\delta_{c_j}(1) \non \\&+\Tr(M \rho_{\theta}(x_j)) \Delta_j(0,1)\Bigr) \biggr]\non \\
    &=\Ebb_{p^{\Ascr}(c,x)} \Ebb_{p({\sigma})}\biggl[\sup_{\substack{\theta\in \Theta\\ 0\leq M \leq I}}   \Tr\Bigl(M \sum_{j=1}^{{N^{\Ascr}}} \frac{ \sigma_j}{\sqrt{N^{\Ascr}}}\Delta_j(0,1)\rho_{\theta}(x_j)\Bigr) \biggr]\label{eq:ub_app_1}\\
    & \stackrel{(a)}{\leq}n\Ebb_{p^{\Ascr}(c,x)} \underbrace{\Ebb_{p({\sigma})}\biggl[\sup_{\theta\in \Theta}   \biggl \lVert \sum_{j=1}^{N^{\Ascr}} \sigma_j \frac{1}{\sqrt{N^{\Ascr}}}\Delta_j(0,1)\rho_{\theta}(x_j) \biggr \rVert \biggr]}_{\mathfrak{R}^{\Ascr}_{\Theta}} \label{eq:4}.
\end{align}The inequality in  $(a)$ follows by using the inequality $\Tr(AB)\leq \Tr(A) \lVert B \rVert$ for $A \geq 0$ with $\lVert \cdot \rVert$ denoting the operator norm (or Schatten-$\infty$ norm) \cite{watrous2018theory}, and noting that  $\Tr(M) \leq n$, where $n$ is the (finite) dimension of the Hilbert space. 

We now  obtain an upper bound on the term $\mathfrak{R}^{\Ascr}_{\Theta}$ in \eqref{eq:4}. To this end,  we proceed by obtaining an $\epsilon$-cover \cite{du2022efficient}, for $\epsilon>0$, of the set of density matrices \begin{align}
 \mathcal{F}=  \biggl \lbrace \rho_{\theta}(\cdot)&=  U(\theta,\cdot) \vert 0\rangle \langle 0 \vert U(\theta,\cdot)^{\dag} \biggl\lvert \non \\&  U(\theta,\cdot)=\Bigl(\prod_{l=1}^L U_l(\theta_l)S_l(\cdot)\Bigr),\theta \in \Theta\biggr \rbrace
\end{align} that describe the PQC in \eqref{eq:VQC}, in terms of the operator norm distance. We adopt the notation in \cite{du2022efficient} to define $\epsilon$-covers. The evaluation of $\epsilon$-cover of $\mathcal{F}$ can be done in two steps. First, we obtain an $\tilde{\epsilon}=\epsilon/L$-cover $\Cscr(U(2^k),\tilde{\epsilon},\lVert \cdot \rVert)$, in terms of the operator norm distance  $\lVert \cdot \rVert$, of the space $U(2^k)$ of all $k$-qubit unitary operators. Lemma 1 of \cite{barthel2018fundamental} ensures existence of such a cover with cardinality $|\Cscr(U(2^k),\tilde{\epsilon},\lVert \cdot \rVert )| \leq (7/\tilde{\epsilon})^{2^{2k}}$. Second, we consider the set,
\begin{align*}
    \tilde{\Fscr}=\biggl \lbrace \tilde{\rho}(\cdot)&=\tilde{U}(\cdot) \vert 0\rangle \langle 0 \vert \tilde{U}(\cdot)^{\dag} \biggl| \non \\&\tilde{U}(\cdot)=\Bigl(\prod_{l=1}^L \tilde{U_l}S_l(\cdot)\Bigr), \tilde{U}_l \in \Cscr(U(2^k), \tilde{\epsilon},\lVert \cdot \rVert)\biggr \rbrace.
\end{align*} Then, for every density-valued function $\rho_{\theta}(\cdot) \in \mathcal{F}$, there exists a counterpart $\tilde{\rho}_{\theta}(\cdot) \in \tilde{\Fscr}$ such that \cite{du2022efficient} \begin{align}
    \lVert \rho_{\theta}(\cdot)-\tilde{\rho}_{\theta}(\cdot)\rVert &= \lVert U(\theta,\cdot) \vert 0\rangle \langle 0 \vert U(\theta,\cdot)^{\dag}- \tilde{U}(\cdot) \vert 0\rangle \langle 0 \vert \tilde{U}(\cdot)^{\dag}\rVert \non \\
    &\leq \sum_{l=1}^L \lVert U_l(\theta_l,\cdot) - \tilde{U}_l(\cdot)\rVert  \\
    &\leq \tilde{\epsilon}L=\epsilon,
\end{align} where the first inequality follows by using triangle inequality and by the unitarily invariance of Schatten $p$-norms \cite{watrous2018theory}. Consequently, we denote  $\mathcal{C}(\mathcal{F},\epsilon,\lVert \cdot \rVert))=\tilde{\Fscr}$ as an $\epsilon$-cover, for $\epsilon>0$,  of the space $\mathcal{F}$ with respect to operator norm distance. 

We thus have the inequalities
\begin{align}
   \mathfrak{R}^{\Ascr}_{\Theta}&=\Ebb_{p({\sigma})}\biggl[\sup_{\theta\in \Theta}   \biggl \lVert \sum_{j=1}^{N^{\Ascr}}  \frac{\sigma_j \Delta_j(0,1)}{\sqrt{N^{\Ascr}}}(\rho_{\theta}(x_j)-\tilde{\rho}_{\theta}(x_j)+\tilde{\rho}_{\theta}(x_j)) \biggr \rVert \biggr]\non \\
    & \stackrel{(a)}{\leq} \Ebb_{p({\sigma})}\biggl[\sum_{j=1}^{N^{\Ascr}}  \frac{1}{\sqrt{N^{\Ascr}}} \sup_{\theta\in \Theta}   \biggl \lVert \rho_{\theta}(x_j)-\tilde{\rho}_{\theta}(x_j)\biggr \rVert \biggr] \non \\& + \Ebb_{p({\sigma})}\biggl[\sup_{\theta\in \Theta}   \biggl \lVert \sum_{j=1}^{N^{\Ascr}} \sigma_j \frac{1}{\sqrt{N^{\Ascr}}}\Delta_j(0,1)\tilde{\rho}_{\theta}(x_j) \biggr \rVert \biggr]\non \\
    & \stackrel{(b)}{\leq} \sqrt{N^{\Ascr}}\epsilon + \Ebb_{p({\sigma})}\biggl[\sup_{\theta\in \Theta}   \biggl \lVert \sum_{j=1}^{N^{\Ascr}} \sigma_j \frac{1}{\sqrt{N^{\Ascr}}}\Delta_j(0,1)\tilde{\rho}_{\theta}(x_j) \biggr \rVert_2 \biggr] \label{eq:111},
\end{align} where the inequality in $(a)$ follows from the triangle inequality of operator norm, and the inequality in $(b)$ follows from the use of $\epsilon$-cover of $\mathcal{F}$, and the monotonicity of Schatten $p$-norms, where by $\lVert A \rVert_p \leq \lVert A \rVert_q$ for $1 \leq q \leq p \leq \infty$ \cite{watrous2018theory}. Specifically, for $p=2$, we have  \ie $\lVert A \rVert_2=\sqrt{\Tr(AA^{\dag})}$. Taking $A=\sum_{j=1}^{N^{\Ascr}} \sigma_j \Delta_j(0,1)\tilde{\rho}_{\theta}(x_j)$, we have 
\begin{align}\Tr(AA^{\dag})&=\Tr(A^2)=\sum_j\Tr(\tilde{\rho}_{\theta}(x_j)^2 ) \non \\&+\sum_j \sum_{j'\neq j}\sigma_j \sigma_{j'}\Delta_{j}(0,1)\Delta_{j'}(0,1)\Tr(\tilde{\rho}_{\theta}(x_j)\tilde{\rho}_{\theta}(x_{j'}))\non \\&=D(\tilde{\rho}_{\theta})+C(\tilde{\rho}_{\theta}). \label{eq:tracesquare}\end{align} Subsequently, we get that
\begin{align}
   &  \sup_{\theta \in \Theta} \biggl \lVert \sum_{j=1}^{N^{\Ascr}} \sigma_j \frac{1}{\sqrt{N^{\Ascr}}}\Delta_j(0,1)\tilde{\rho}_{\theta}(x_j) \biggr \rVert_2 \non \\
   &= \frac{1}{\sqrt{N^{\Ascr}}} \sup_{\hat{\rho} \in \Cscr(\Fscr,\epsilon,\lVert \cdot \rVert)} \sqrt{C(\hat{\rho})+D(\hat{\rho}) } \non \\
   & \leq \frac{1}{\sqrt{N^{\Ascr}}} \sqrt{\sup_{\hat{\rho}\in \in \Cscr(\Fscr,\epsilon,\lVert \cdot \rVert) }C(\hat{\rho})+\sup_{\hat{\rho}\in \in \Cscr(\Fscr,\epsilon,\lVert \cdot \rVert) } \sum_j\Tr(\hat{\rho}(x_j)^2 )}\non\\
   & \leq \frac{1}{\sqrt{N^{\Ascr}}}  \sqrt{\sum_{\hat{\rho}\in \in \Cscr(\Fscr,\epsilon,\lVert \cdot \rVert) }C(\hat{\rho})+\sup_{\hat{\rho} \in \in \Cscr(\Fscr,\epsilon,\lVert \cdot \rVert) } \sum_j\Tr(\hat{\rho}(x_j)^2 )}.\non
\end{align} Taking expectation with respect to $p({\sigma})$, and using Jensen's inequality then gives the inequality
\begin{align}
    &  \Ebb_{p({\sigma})}\biggl[\sup_{\theta \in \Theta} \biggl \lVert \sum_{j=1}^{N^{\Ascr}} \sigma_j \frac{1}{\sqrt{N^{\Ascr}}}\Delta_j(0,1)\tilde{\rho}_{\theta}(x_j) \biggr \rVert_2 \biggr] \non \\
    & \leq \frac{1}{\sqrt{N^{\Ascr}}}  \sqrt{\sum_{\hat{\rho} \in \in \Cscr(\Fscr,\epsilon,\lVert \cdot \rVert)}\Ebb_{p({\sigma})}[C(\hat{\rho})]+\sup_{\hat{\rho} \in \in \Cscr(\Fscr,\epsilon,\lVert \cdot \rVert)} \sum_j\Tr(\hat{\rho}(x_j)^2 )}\non \\
    &= \frac{1}{\sqrt{N^{\Ascr}}} \sqrt{\sup_{\hat{\rho} \in \in \Cscr(\Fscr,\epsilon,\lVert \cdot \rVert)}\sum_j\Tr(\hat{\rho}(x_j)^2 )},\non
\end{align} where the last equality follows since Rademacher variables are i.i.d and mean zero, whereby $\Ebb_{p({\sigma})}[C(\hat{\rho})]=0$.
Finally, we then have that for any $\epsilon>0$,
\begin{align}
    \mathfrak{R}_{\Theta,\Mscr}^{\Ascr} &\leq n\sqrt{N^{\Ascr}}\epsilon + \frac{n}{\sqrt{N^{\Ascr}}}\Ebb_{p^{\Ascr}(x)}\biggl[\sqrt{\sum_j \sup_{\hat{\rho} \in \in \Cscr(\Fscr,\epsilon,\lVert \cdot \rVert)} \Tr(\hat{\rho}(x_j)^2 )} \biggr] \non \\
& \leq n\sqrt{N^{\Ascr}}\epsilon+n\sqrt{\Ebb_{p^{\Ascr}(x)}[\sup_{\hat{\rho} \in \in \Cscr(\Fscr,\epsilon,\lVert \cdot \rVert)} \Tr(\hat{\rho}(x)^2 )]}.
\end{align} Taking $\epsilon$ arbitrarily small, we get that
\begin{align}
    \mathfrak{R}_{\Theta,\Mscr}^{\Ascr} &\leq n\sqrt{\Ebb_{p^{\Ascr}(x)}[\sup_{\hat{\rho} \in \in \Cscr(\Fscr,\epsilon,\lVert \cdot \rVert)} \Tr(\hat{\rho}(x)^2 )]}.
\end{align}

\subsection{ One-Time Data Encoding Ansatz}
In this section, we consider PQC's of the form
\begin{align}
    U(\theta,x)= \underbrace{\prod_{l=1}^L U_l(\theta_l)}_{:=U(\theta)}S(x),
\end{align} whereby quantum states are of the form
\begin{align}
    \rho_{\theta}(x)= U(\theta)S(x)\vert 0 \rangle \langle 0 \vert S(x)^{\dag} U(\theta)^{\dag}.
\end{align} For such quantum states, the Rademacher complexity in \eqref{eq:ub_app_1} can be tightened as 
\begin{align}
    &\mathfrak{R}^{\Ascr}_{\Theta,\Mscr} \leq \Ebb_{p^{\Ascr}(c,x)} \Ebb_{p({\sigma})}\biggl[\sup_{\theta\in \Theta}   \biggl \lVert \sum_{j=1}^{N^{\Ascr}} \sigma_j \frac{1}{\sqrt{N^{\Ascr}}}\Delta_j(0,1)\rho_{\theta}(x_j) \biggr \rVert_1 \biggr]\non \\
    &=\Ebb_{p^{\Ascr}(c,x)} \Ebb_{p({\sigma})}\biggl[\sup_{\theta\in \Theta}   \biggl \lVert U(\theta)\biggl(\sum_{j=1}^{N^{\Ascr}} \sigma_j \frac{1}{\sqrt{N^{\Ascr}}}\Delta_j(0,1)\kappa(x_j)\biggr)U(\theta)^{\dag} \biggr \rVert_1 \biggr] \non \\
    &\stackrel{(a)}{=}\Ebb_{p^{\Ascr}(c,x)} \Ebb_{p({\sigma})}\biggl[\sup_{\theta\in \Theta}   \biggl \lVert \sum_{j=1}^{N^{\Ascr}} \sigma_j \frac{1}{\sqrt{N^{\Ascr}}}\Delta_j(0,1)\kappa(x_j)\biggr \rVert_1 \biggr] \non \\
    &=\Ebb_{p^{\Ascr}(c,x)} \Ebb_{p({\sigma})}\biggl[  \biggl \lVert \underbrace{\sum_{j=1}^{N^{\Ascr}} \sigma_j \frac{1}{\sqrt{N^{\Ascr}}}\Delta_j(0,1)\kappa(x_j)}_{:=B}\biggr \rVert_1 \biggr] \non \\
    & \stackrel{(b)}{\leq} \Tr \biggl(\sqrt{\Ebb_{p^{\Ascr}(c,x)} \Ebb_{p({\sigma})}[BB^{\dag}]} \biggr)\non \\
    &=\Tr \biggl(\sqrt{\Ebb_{p^{\Ascr}(x)} [\kappa(x)^2]} \biggr),
\end{align} where we have used $\kappa(x_j)=S(x_j)\vert 0 \rangle \langle 0 \vert S(x_j)^{\dag}$. The first inequality follows from tracial Matrix H\'olders inequality that $|\Tr(A^{\dag}B)|\leq \lVert A^{\dag} \rVert_{\infty} \lVert B \rVert_1$  and noting that $\lVert M \rVert_{\infty}\leq 1$. The equality in $(a)$ follows since trace norm is unitarily invariant and inequality in $(b)$ is due to \cite[Lemma 1]{banchi2021generalization}. The last equality can be seen by evaluating $BB^{\dag}$ as in \eqref{eq:tracesquare}.

\section{Details of Example}
For the example considered in Section~\ref{sec:example}, the Pauli-X rotation is defined as
\begin{align}
R_X(x)=
    \begin{bmatrix}
    \cos(x/2)& -i \sin(x/2)\\
    i \sin(x/2) & \cos(x/2)
    \end{bmatrix},
\end{align} 
and the general rotation $\Rot_{\theta}$ is defined as
\begin{align*}
   \Rot_{\theta}=\begin{bmatrix} e^{i(-\frac{\theta_1}{2}-\frac{\theta_3}{2})}\cos(\theta_2/2) & -e^{i(-\frac{\theta_1}{2}+\frac{\theta_3}{2})}\sin(\theta_2/2)\\
    e^{i(\frac{\theta_1}{2}-\frac{\theta_3}{2})}\sin(\theta_2/2)& e^{i(\frac{\theta_1}{2}+\frac{\theta_3}{2})}\cos(\theta_2/2)
    \end{bmatrix}.
    \end{align*}

\end{document}